\documentclass[a4paper,10pt]{article}
\usepackage{comment}
\usepackage[utf8]{inputenc}
\usepackage{enumerate}
\usepackage[OT4]{fontenc}
\usepackage{graphicx}
\usepackage{amssymb}
\usepackage{amsmath}
\usepackage{amsfonts}
\usepackage{amsthm}
\usepackage{algorithm}
\usepackage{authblk}
\usepackage{fullpage}
\usepackage[noend]{algpseudocode}
\usepackage{algorithmicx}
\usepackage[textsize=scriptsize]{todonotes}
\usepackage[hidelinks]{hyperref}
\usepackage[capitalise]{cleveref}
\usepackage{thmtools}
\usepackage{thm-restate}
\usepackage{microtype}
\usepackage[shortlabels]{enumitem}
\usepackage{tikz}
\usepackage[all,cmtip]{xy}
\usepackage{caption}
\usepackage{tikzscale}

\usetikzlibrary{backgrounds}
\usetikzlibrary{patterns}
\usetikzlibrary{decorations.pathmorphing}
\usetikzlibrary{decorations.pathreplacing}

\newlength{\xywd}
\newcommand{\xyrightarrow}[2][]{%
  \sbox{0}{$\scriptstyle#1$}%
  \xywd=\wd0
  \sbox{0}{$\scriptstyle#2$}%
  \ifdim\wd0>\xywd \xywd=\wd0 \fi
  \xymatrix@C\dimexpr\xywd+1em\relax{{}\ar[r]^{#2}_{#1}&{}}%
}

\let\originalleft\left
\let\originalright\right
\renewcommand{\left}{\mathopen{}\mathclose\bgroup\originalleft}
\renewcommand{\right}{\aftergroup\egroup\originalright}

\newtheorem{theorem}{Theorem}[section]
\newtheorem{lemma}[theorem]{Lemma}

\newtheorem{definition}[theorem]{Definition}
\newtheorem{corollary}[theorem]{Corollary}

\def\poly{\operatorname{poly}}

\newcommand{\eps}{\epsilon}
\newcommand{\dist}{\delta}

\newcommand{\wei}{w}
\newcommand{\disth}[1]{\ensuremath{\dist^{#1}}}
\newcommand{\len}{\ell}
\newcommand{\rev}[1]{\ensuremath{#1}^{\mathrm{R}}}

\newcommand{\treecol}{\mathcal{T}}
\newcommand{\hops}[1]{\ensuremath{|#1|}}

\newcommand{\tfrom}{\ensuremath{\mathrm{from}}}
\newcommand{\tto}{\ensuremath{\mathrm{to}}}
\newcommand{\treepath}[2]{\ensuremath{{#1}[{#2}]}}
\newcommand{\treedep}[2]{\ensuremath{\mathrm{dep}_{#1}({#2})}}
\newcommand{\prob}{\ensuremath{\mathrm{Pr}}}

\newcommand{\tlink}{\ensuremath{\mathtt{link}}}
\newcommand{\tcut}{\ensuremath{\mathtt{cut}}}
\newcommand{\tdepth}{\ensuremath{\mathtt{depth}}}
\newcommand{\tparop}{\ensuremath{\mathtt{parent}}}

\newcommand{\troot}{\ensuremath{\mathrm{root}}}
\newcommand{\tpar}{\ensuremath{\mathrm{par}}}
\newcommand{\nil}{\ensuremath{\mathbf{nil}}}
\newcommand{\dstr}{\ensuremath{\mathcal{D}}}

\newcommand{\expround}{\mathrm{expround}}

\author[1]{Adam Karczmarz\thanks{Supported by ERC Consolidator
Grant 772346 TUgbOAT and the Polish National Science Centre 2018/29/N/ST6/00757 grant.}}
\author[2]{Jakub Łącki}

\affil[1]{Institute of Informatics, University of Warsaw, Poland}
\affil[ ]{\texttt{a.karczmarz@mimuw.edu.pl} \medskip}

\affil[2]{Google Research, New York, USA}
\affil[ ]{\texttt{jlacki@google.com}}

  \newcommand{\Ot}{\ensuremath{\widetilde{O}}}
\begin{document}

\date{}
	\title{Reliable Hubs for Partially-Dynamic\\All-Pairs Shortest Paths in Directed Graphs}
\maketitle
	
\begin{abstract}
We give new partially-dynamic algorithms for the all-pairs shortest paths problem in weighted directed graphs.
  Most importantly, we give a new \emph{deterministic} incremental algorithm for the problem that handles updates in $\Ot(mn^{4/3}\log{W}/\epsilon)$ total time (where the edge weights are from $[1,W]$) and explicitly maintains a $(1+\epsilon)$-approximate distance matrix. For a fixed $\epsilon>0$, this is the first deterministic partially dynamic algorithm for all-pairs shortest paths in directed graphs, whose update time is $o(n^2)$ regardless of the number of edges.
Furthermore, we also show how to improve the state-of-the-art partially dynamic \emph{randomized} algorithms for all-pairs shortest paths [Baswana et al. STOC’02, Bernstein STOC’13] from Monte Carlo randomized to Las Vegas randomized without increasing the running time bounds (with respect to the $\Ot(\cdot)$ notation).

Our results are obtained by giving new algorithms for the problem of dynamically maintaining \emph{hubs}, that is a set of $\Ot(n/d)$ vertices which~hit a shortest path between each pair of vertices, provided it has hop-length $\Omega(d)$.
We give new~subquadratic deterministic and Las Vegas algorithms for maintenance of hubs under either edge insertions or deletions.
\end{abstract}
\section{Introduction}\label{sec:introduction}
The sampling scheme of Ullman and Yannakakis~\cite{UY91} is a fundamental tool in designing dynamic algorithms for maintaining shortest path distances.
  Roughly speaking, the main idea is that if each vertex of the graph is sampled independently with probability $\Omega(\frac{d \ln n}{n})$, then with high probability\footnote{We say that a probabilistic statement holds \emph{with high (low) probability}. abbreviated w.h.p., if it holds with probability at least $1 - n^{-\beta}$ (at most $n^{-\beta}$, resp.), where $\beta$ is a constant that can be fixed arbitrarily.} the set of the sampled vertices has the following property.
If the shortest path between some vertices $u$ and $v$ contains more than $d$ edges, then this shortest path contains a sampled vertex\footnote{For simplicity, in the introduction we assume that the shortest paths are unique.}.
We call each set having this property a set of \emph{hubs}\footnote{Zwick \cite{Zwick02} uses the name \emph{bridging set} for an analogous concept.
Some works also use the term \emph{hitting set}, but hitting set is a more general notion, which in our paper is used in multiple different contexts.} of that graph.

The fact that one can easily obtain a set of hubs by random sampling is particularly useful for dynamic graph algorithms, since, by tuning constants in the sampling probability, one can assure that the set of hubs remains valid at each step (with high probability), while the graph is undergoing edge insertions and deletions, assuming the total number of updates is polynomial.
This property has been successfully exploited to give a number of dynamic graph algorithms, e.g.~\cite{DBLP:conf/soda/AbrahamCK17, BaswanaHS07, Bernstein16, DBLP:journals/jcss/DemetrescuI06, HenzingerKN14, HenzingerKN14a, Henzinger15, Roditty08, DBLP:journals/algorithmica/RodittyZ11, DBLP:journals/siamcomp/RodittyZ12}.
At the same time, the sampling approach also suffers from two drawbacks.
First, it yields Monte Carlo algorithms, which with some nonzero probability can return incorrect answers.
Second, it relies on the \emph{oblivious adversary} assumption, that is, it requires that the updates to the graph are independent of the randomness used for sampling hubs.
This becomes a substantial issue for problems where the answer to a query is not unique, e.g., for maintaining $(1+\epsilon)$-approximate distances or maintaining the shortest paths themselves (i.e. not just their lengths).
In a typical case, the choice of the specific answer to a query depends on the randomness used for vertex sampling, which in turn means that in each answer to a query the data structure is revealing its randomness.
Hence, if the following updates to the data structure depend on the specific values returned by the previous queries, the oblivious adversary assumption is not met.

In this paper we attempt to address both these issues.
We study the dynamic maintenance of \emph{reliable} hubs, that is we show how to maintain hubs using an algorithm that does not err, even with small probability.
In addition, in the incremental setting we give an algorithm that maintains hubs deterministically.
While the algorithms are relatively straightforward for unweighted graphs, making them also work in the weighted setting is a major challenge, which we manage to overcome.
We then show how to take advantage of our results on reliable hubs to obtain improved algorithms for the problem of maintaining all-pairs shortest paths in directed graphs.
In particular, we give a faster deterministic incremental algorithm and show how to improve the state-of-the-art decremental algorithms from Monte Carlo to Las Vegas.

\subsection{Our Contribution}

We study the problem of maintaining \emph{reliable} hub sets in the partially dynamic setting.
For the description, let us first assume the case when the graph is unweighted.
Our first observation is that
one can deterministically maintain the set of hubs $H_d$ under edge insertions in $\Ot(nmd)$ total time.
To that end, we observe that after an edge $uw$ is inserted, we may ensure the set of hubs $H_d$ is valid by extending it with both $u$ and $w$. This increases the size of $H_d$, and hence we have to periodically discard all the hubs and recompute them from scratch.
 
The deterministic computation of hubs has been studied before.
For unweighted digraphs, King~\cite{King99} showed how to compute a hub set $H_d$ of size $\Ot\left(\frac{n}{d}\right)$ in $\tilde{O}(n^2)$ time.
The algorithm, given shortest path trees up to depth $d$ from all vertices $v\in V$, 
computes a \emph{blocker-set}~\cite{King99} of these trees. (A~blocker-set $S$ of a rooted tree is a set such that, for each path from the root to a leaf of length $d$, that path contains a vertex of $S$ distinct from the root.)
Hence, if we work on unweighted graphs, in order to keep the set $H_d$ valid and relatively small, we can maintain shortest path
trees up to depth $d$ from all vertices using the Even-Shiloach algorithm \cite{EvenS81} in $O(nmd)$
total time, and recompute $H_d$ using King's algorithm every $\Ot(\frac{n}{d})$
insertions.
The total time needed for maintaining the hubs is therefore $\Ot(nmd)$.

Furthermore, we also show how to maintain reliable hubs in a decremental setting.
Suppose our goal is to compute a set of hubs that is guaranteed to be valid,
which clearly is not the case for the sampled hubs of~\cite{UY91}.
We show that if shortest path trees up to depth $d$ are maintained
using \emph{dynamic tree} data structures \cite{SleatorT83, Tarjan97}, one can recompute
a certainly-valid set $H_d$ in $\Ot\left(\frac{n^2}{d}\right)$ time using a Las Vegas algorithm.
To this end observe that one can deterministically \emph{verify} if a set $B\subseteq V$
is a blocker-set of $n$ shortest path trees up to depth $d$
in $\Ot(n\cdot|B|)$ time.
Therefore, a hub set $H_d$ can be found by combining the approaches
of \cite{UY91} and \cite{King99}: we may sample candidate hub sets
of size $\Ot(\frac{n}{d})$ until a blocker-set of the trees is found.
The number of trials is clearly constant with high probability.

We further extend this idea and show that the information whether $B$
is a blocker-set of a collection of $n$ shortest path trees up to depth $d$
can be maintained subject to the changes to these trees with only polylogarithmic overhead.
Consequently, we can detect 
when the sampled hub set $H_d$ (for any $d$) ceases
to be a valid hub set in $\Ot(nmd)$ total time.
The algorithm may make one-sided error (i.e., say that $H_d$ is no longer a valid hub set when it is actually still good),
but the probability of an error is low  if we assume that the update sequence does not depend on our random bits.
Subsequently we show how to extend this idea to improve the total update time to $\Ot(nm)$.
Assume we are given a valid $d$-hub set $H_d$.
We prove that in order to verify whether $H_{6d}$ is a valid $6d$-hub set, it suffices to check whether it hits sufficiently long paths between the elements of $H_d$.
We use this observation to maintain a family of reliable hub sets $H_{1},H_{6},\ldots,H_{6^i},\ldots,H_{6^k}$ (where $6^k\leq n$)
under edge deletions (or under edge insertions) in $\Ot(nm)$ total time.
Using that, we immediately improve the state-of-the-art decremental APSP
algorithms of Baswana et al. \cite{BaswanaHS07} (for the exact unweighted case) and
Bernstein~\cite{Bernstein16} (for the $(1+\eps)$-approximate case)
from Monte Carlo to Las Vegas (but still assuming an oblivious adversary) by only adding a polylogarithmic factor
to the total update time bound.

\paragraph{Generalization to weighted digraphs.}

Adapting the reliable hub sets maintenance (for both described approaches: the incremental
one and sample/verify)
to \emph{weighted} digraphs
turns out to be far from trivial.
This is much different from the sampling approach of Ullman and Yannakakis \cite{UY91}, which works regardless of whether
the input graph is weighted or not.
The primary difficulty is maintaining all shortest paths consisting of up to $d$ edges.
While in the unweighted case the length of a path is equal to the number of edges on this path, this is no longer true in the weighted case.

To bypass this problem 
we first relax our definition of hubs. For 
each $u,v\in V$ we require that some $(1+\eps)$-approximate shortest $u\to v$ path contains a hub on each subpath consisting of at least $d+1$ edges.
Next, we show that running King's blocker-set algorithm on a set of $(1+\eps)$-approximate shortest path trees up to depth\footnote{In such a tree (see Definition~\ref{def:appr-tree}), which is a subgraph of $G$,
for all $v\in V$, the path from the source $s$ to $v$ has length not exceeding $(1+\eps)$ times
the length of a shortest out of $s\to v$ paths in $G$ that use no more than $d$ edges; however the tree path itself can have arbitrary number of edges.} $d$
from all vertices of the graph
yields a hub set that hits paths
approximating the true shortest paths within a factor of $(1+\eps)^{\Theta(\log{n})}$.
Note that a collection of such trees can be maintained in $\Ot(nmd\log{W}/\eps)$ total
time subject to edge insertions, using Bernstein's $h$-SSSP algorithm~\cite{Bernstein16}
with $h=d$.

The $\Theta(\log n)$ exponent in the approximation ratio comes from the following difference between the weighted and unweighted case.
In a $(1+\eps)$-approximate shortest path tree up to depth $d$, the length of a $u\to v$ path is no more than $(1+\eps)$-times the length of the shortest $u\to v$ path
in~$G$ that uses at most $d$ edges. However, the $u\to v$ path in the tree might consist of any number of edges, in particular very few.
Pessimistically, all these trees have depth $o(d)$ and their blocker-set is empty, as there is no path of hop-length $\Omega(d)$ that we need to hit.
Note that this is an inherent problem, as the fact that we can find a small blocker-set in the unweighted case relies on the property that we want it to hit paths of $\Omega(d)$ edges.

Luckily, a deeper analysis shows that our algorithm can still approximate the length of a $s \to v$ path.
Roughly speaking, we split the $s \to v$ path $P$ into two subpaths of $d/2$ edges.
If each of these two subpaths are approximated in the $h$-SSSP data structures by paths of less than $d/4$ edges, we replace the $P$ by the concatenation of the two approximate paths from the $h$-SSSP data structures.
This way, we get a path that can be longer by a factor~of~$(1+\eps)$, but whose hop-length is twice smaller.
By repeating this process $O(\log n)$ times we obtain a path of constant hop-length whose length is at most $(1+\eps)^{\Theta(\log n)}$ larger than the length of~$P$.
The overall approximation ratio is reduced to $(1+\eps)$ by scaling $\eps$ by a factor of $\Theta(\log n)$.

\paragraph{Deterministic incremental all-pairs shortest paths.}
We now show how to apply our results on reliable hubs to obtain an improved algorithm for incremental all-pairs shortest paths problem in weighted digraphs.
 We give a deterministic incremental algorithm maintaining all-pairs $(1+\eps)$-approximate distance estimates in $\Ot(mn^{4/3}\log{W}/\eps)$ total time.


Let us now give a brief overview of our algorithm in the unweighted case.
First, we maintain the set of hubs $H_d$ under edge insertions as described above in $\Ot(nmd)$
total time.
Second, since the set $H_d$ changes and each vertex of the graph may eventually
end up in $H_d$, we cannot afford maintaining shortest path
trees from all the hubs (which is done in most algorithms that use hubs).
Instead, we use the folklore $\Ot(n^3/\eps)$ total time incremental $(1+\eps)$-approximate APSP algorithm \cite{Bernstein16, King99}
to compute distances between the hubs.
Specifically, we run it on a graph whose vertex set is $H_d$ and whose edges represent shortest paths between hubs of hop-lengths at most~$d$.
These shortest paths are taken from the shortest path trees up to depth $d$ from all $v\in V$ that are required
for the hub set maintenance.
We reinitialize the algorithm  each time the set $H_d$ is recomputed.
This allows us to maintain approximate pairwise distances
between the hubs at all times in $\Ot\left(m\left(n/d\right)^2/\eps\right)$ total time.

Finally, we show how to run a dynamic algorithm on top of a \emph{changing} set of hubs
by adapting the shortcut edges technique of Bernstein~\cite{Bernstein16}.
Roughly speaking, the final estimates are maintained using $(1+\eps)$-approximate
shortest path trees \cite{Bernstein16} up to depth $O(d)$ from all vertices $v$ on graph $G$ augmented with shortcuts from $v$ to $H_d$ and from $H_d$ to $v$.
This poses some technical difficulties as the set of shortcuts is undergoing both insertions (when a hub is added) and deletions (when the entire set of hubs is recomputed from scratch).
However, one can note that in the incremental setting the shortcuts
that no longer approximate the distances between their endpoints do not
break the approximation guarantee of our algorithm.
Eventually, we use shortcuts between all pairs of vertices of $G$
but only some of them are guaranteed (and sufficient)
to be up to date at any time.
The total time cost of maintaining this component is $\Ot(nmd/\eps)$.
Setting $d=\Ot(n^{1/3})$ gives the best update time.

It is natural to wonder if this approach could be made to work in the decremental setting. There are two major obstacles. First, it is unclear whether one can deterministically maintain a valid set of hubs under deletions so that only $O(1)$ vertices (in amortized sense) are added to the hub set after each edge deletion. Note that in extreme cases, after a single edge deletion the set of hubs may have to be extended with polynomially many new vertices. Second, all algorithms using the above approach of introducing shortcuts from and to hubs
also maintain a decremental shortest path data structure on a graph consisting of the edges of the original graph and shortcut edges representing distances between the hubs. If hubs were to be added, the graph maintained by the data structure would undergo both insertions (of shortcuts) and deletions (of edges of the original graph) which would make this a much harder, fully dynamic problem. Some earlier works dealt with a similar issue by ignoring some ``inconvenient" edge insertions~\cite{Henzinger16} or showing that the insertions are well-behaved~\cite{BR11}.
However, these approaches crucially depended on the graph being undirected.

\subsection{Related Work}
The dynamic graph problems on digraphs are considerably harder than their counterparts on undirected graphs.
An extreme example is the dynamic reachability problem, that is, transitive closure on directed graphs, and connectivity on undirected graphs.
While there exist algorithms for undirected graphs with polylogarithmic query and update times~\cite{Holm01,WN13,Thorup00,HK95,Kapron13}, in the case of directed graphs the best known algorithm with polylogarithmic query time has an update time of $O(n^2)$~\cite{Sankowski04, Demetrescu00, Roditty08}.
In addition, a combinatorial algorithm with an update time of $O(n^{2-\epsilon})$ is ruled out under Boolean matrix multiplication conjecture~\cite{Abboud14}.

In 2003, in a breakthrough result Demetrescu and Italiano gave a fully dynamic, exact and deterministic algorithm for APSP in weighted directed graphs~\cite{DemetrescuI04}. The algorithm handles updates in $\Ot(n^2)$ amortized time and maintains the distance matrix explicitly. The bound of $O(n^2)$ is a natural barrier as a single edge insertion or deletion may change up to $\Omega(n^2)$ entries in the distance matrix. For dynamic APSP in digraphs there exists faster algorithms with polylogarithmic query time, all of which work in incremental or decremental setting:
\begin{itemize}
	\item Ausiello et al.~\cite{AusielloIMN90} gave a deterministic incremental algorithm for exact distances in unweighted digraphs that handles updates in $\Ot(n^3)$ total time.
\item Baswana et al.~\cite{BaswanaHS07} solved the same problem in the decremental setting with a Monte Carlo algorithm with $\Ot(n^3)$ total update time.
\item Bernstein~\cite{Bernstein16} gave a Monte Carlo algorithm for $(1+\epsilon)$-approximate distances in weighted graphs (with weights in $[1,W]$) with $\Ot(nm\log{W}/\eps)$ total update time. 
The algorithm works both in the incremental and decremental setting. 
\item Finally, deterministic partially-dynamic (both incremental and decremental) algorithms for APSP in directed graphs with $\Ot(n^3\log{W}/\eps)$ total
	update time can be obtained by combining the results of~\cite{King99} and~\cite{Bernstein16}.
\end{itemize}
The algorithms of Baswana et al.~~\cite{BaswanaHS07} and Bernstein~\cite{Bernstein16} both use sampled hubs and thus require the \emph{oblivious adversary} assumption.
We highlight that
in the class of deterministic algorithms, the best known results have total update time $\Ot(n^3)$~\cite{AusielloIMN90, Bernstein16}, even if we only consider sparse unweighted graphs in incremental or decremental setting and allow $(1+\epsilon)$ approximation.
In the incremental setting, for not very dense graphs, when $m=O(n^{5/3-\eps})$, our algorithm improves this bound to $\Ot(mn^{4/3})$.

\paragraph{Organization of the paper.}
In Section~\ref{sec:preliminaries} we fix notation, review some of the existing tools that we use and give a formal definition of hubs.
Section~\ref{sec:sparse-incremental} describes the hub set maintenance for incremental unweighted digraphs
and our $(1+\eps)$-approximate incremental algorithm for sparse graphs.
In Section~\ref{sec:verification} we show a faster Las Vegas algorithm for computing
reliable hubs and further extend it to maintain reliable hub sets in the partially dynamic setting.
There we also sketch how to use it in order to 
to improve the state-of-the-art decremental APSP algorithms from Monte Carlo to Las Vegas randomized.
Finally, in Section~\ref{sec:weighted} we show how to adapt the hub set maintenance
algorithms of Sections~\ref{sec:sparse-incremental}~and~\ref{sec:verification}, so that they work on weighted graphs.
\section{Preliminaries}\label{sec:preliminaries}

In this paper we deal with \emph{directed} graphs.
We write $V(G)$ and $E(G)$ to denote the sets of vertices and edges of $G$, respectively.
A graph $H$ is a \emph{subgraph} of $G$, which we denote by $H\subseteq G$, if
and only if $V(H)\subseteq V(G)$ and $E(H)\subseteq E(G)$.
We write $uv\in E(G)$ when referring to edges of $G$ and use $\wei_G(uv)$
to denote the weight of $uv$.
If $G=(V,E)$ is unweighted, then $\wei_G(e) = 1$ for each $e\in E$.
For weighted graphs,
$\wei_G(e)$ can be any real number from the interval $[1,W]$.
For simplicity, in this paper we assume that $W$ is an input parameter
given beforehand.
If $uv\notin E$, we assume $\wei_G(uv)=\infty$.
We define the union $G\cup H$ to be the graph $(V(G)\cup V(H),E(G)\cup E(H))$
with weights $\wei_{G\cup H}(uv)=\min(\wei_G(uv),\wei_H(uv))$ for each $uv\in E(G\cup H)$.
For an edge $e=uv$, we write $G+e$ to denote $(V(G) \cup \{u,v\}, E(G) \cup \{e\})$. The \emph{reverse graph} $\rev{G}$ is defined as $(V(G),\{xy:yx\in E(G)\})$
and $\wei_{\rev{G}}(xy)=\wei_G(yx)$.

A sequence of edges $P=e_1\ldots e_k$, where $k\geq 1$ and $e_i=u_iv_i\in E(G)$, is called 
a $u\to v$ path in $G$ if $u=u_1$, $v_k=v$ and $v_{i-1}=u_i$ for each $i=2,\ldots,k$.
We sometimes view a path $P$ in $G$ as a subgraph of $G$ with vertices $\{u_1,\ldots,u_k,v\}$
and edges $\{e_1,\ldots,e_k\}$
and write $P\subseteq G$.
The \emph{hop-length} $\hops{P}$ is defined as $\hops{P}=k$.
The \emph{length} of the path $\len(P)$ is defined as $\len(P)=\sum_{i=1}^k\wei_G(e_i)$.
If $G$ is unweighted, then we clearly have $\hops{P}=\len(P)$.
For convenience, we sometimes consider a single edge $uv$ a path of hop-length $1$.
It is also useful to define a length-$0$ $u\to u$ path to be the graph $(\{u\},\emptyset)$.
If $P_1$ is a $u \to v$ path and $P_2$ is a $v \to w$ path, we denote by $P_1\cdot P_2$ (or simply $P_1P_2$) a path $P_1 \cup P_2$ obtained by concatenating $P_1$ with $P_2$.

A digraph $T$ is called an \emph{out-tree over $V$ rooted in $r$} if $v\in V(T)\subseteq V$,
$|E(T)|=|V(T)|-1$ and for all $v\in V(T)$ there is a unique
path $\treepath{T}{v}$ from $r$ to $v$.
The depth $\treedep{T}{v}$ of a vertex $v\in V(T)$ is defined as $\hops{\treepath{T}{v}}$.
The depth of $T$ is defined as $\max_{v\in V(T)}\{\treedep{T}{v}\}$.
Each non-root vertex of an out-tree has exactly one incoming
edge.
For $v\in V(T)\setminus\{r\}$ we call the other endpoint of the incoming edge of $v$
the \emph{parent} $v$ and write $\tpar_T(v)$ when referring to it.

The \emph{distance} $\dist_G(u,v)$ between the vertices $u,v\in V(G)$ is the length
of the shortest $u\to v$ path in $G$, or $\infty$, if no $u\to v$ path
exists in $G$.
We define $\disth{k}_G(u,v)$ to be the length of the shortest path from $u$ to $v$ among paths of at most $k$ edges. Formally, $\disth{k}_G(u,v)=\min\{\len(P):u\to v=P\subseteq G\text{ and }\hops{P}\leq k\}$.
We sometimes omit the subscript $G$ and write $\wei(uv)$, $\dist(u,v)$, $\disth{k}(u,v)$ etc.
instead of $\wei_G(u,v)$, $\dist_G(u,v)$, $\disth{k}_G(u,v)$, etc., respectively.

We say that a graph $G$ is \emph{incremental}, if it only undergoes edge insertions
and edge weight decreases. Similarly, we say that $G$ is \emph{decremental}
if it undergoes only edge deletions and edge weight increases.
We say that $G$ is \emph{partially dynamic} if it is either incremental
or decremental.
For a dynamic graph $G$ we denote by $n$ the maximum value of $|V|$ and by $m$ the maximum
value of $|E|$ throughout the whole sequence of updates.

When analyzing $(1+\eps)$-approximate algorithms, we assume $0<\eps<1$
and $1/\eps=\poly{n}$.\footnote{If this was not the case, we had better use exact fully-dynamic APSP algorithm \cite{DemetrescuI04} instead.}

We denote by $\Delta$ the total number of updates a dynamic graph $G$ is
subject to.
If $G$ is unweighted, then clearly $\Delta\leq m$ and in fact
we assume $\Delta=m$, which allows us to simplify the analyses.
For weighted digraphs, on the other hand, since the total number
of weight increases/decreases that an edge is subject to is unlimited,
$\Delta$ may be much larger than $m$.
As a result, it has to be taken into account when analyzing the efficiency
of our algorithms.

We call a partially-dynamic $(1+\eps)$-approximate APSP problem on weighted
graphs \emph{restricted} if the edge weights of $G$ are of the form $(1+\eps)^i$ for $i \in [0, \lceil \log_{1+\eps}{W}\rceil] \cap \mathbb{N}$
at all times and additionally each update is required to actually change the
edge set or change the weight of some existing edge.
Consequently, observe that in the restricted problem we have
$\Delta\leq m\cdot (\lceil\log_{1+\eps}{W}\rceil+2)$.
In the following we concentrate on the restricted problem.
This is without much loss of generality as proved below.
\begin{restatable}{lemma}{opreduction}
\label{lem:op-reduction}
Let $\mathcal{A}$ be an algorithm solving 
the restricted $(1+\eps)$-approximate all-pairs shortest paths problem on a weighted
digraph in $O(T(n,m,W,\eps,\Delta))$ total time, where $T(n,m,W,\eps,\Delta)=\poly{(n,m,W,\eps,\Delta)}$.
  Then, there exists an algorithm solving the corresponding general problem
  (i.e., with arbitrary real edge-weights from the range $[1,W]$)
  in
  $O(T(n,m,W,\eps,m\log{W}/\eps)+\Delta)$ total time.
\end{restatable}

\begin{proof}
	Let $\eps'=\eps/4$.
	Let $G'$ have the same vertices and edges as $G$
	and let $$\wei_{G'}(uv)=\expround_{(1+\eps')}{\wei_G(uv)}$$ for each $uv\in E$, where $\expround_a(x) := x^{\lceil \log_a x \rceil}$.
  Note that for any $u,v\in V$, $\dist_{G'}(u,v)\leq (1+\eps')\dist_{G}(u,v)$.
  
  Let us run an instance $A$ of the $(1+\eps')$-approximate algorithm $\mathcal{A}$ on $G'$.
  $A$ yields
  $(1+\eps')^2 \leq (1+\eps)$-approximate distances in $G$.
  However, not all updates to $G$ are passed to $A$.
  if $\wei_{G'}(uv)$ does not change as a result of
  an update to $\wei_G(uv)$, it can be ignored from the point
  of view of $A$ in $O(1)$ time.
  On the other hand, if a change to $\wei_G(uv)$ makes
  $\wei_{G'}(uv)$ change, such an update is passed to $A$.
  However, for each edge $uv\in E$ this can clearly happen
  at most $O(\log_{1+\eps'}{W})=O(\log{W}/\eps')=O(\log{W}/\eps)$ times.

  Hence, the total update time is indeed 
$O(T(n,m,W,\eps,m\log{W}/\eps)+\Delta)$.
\end{proof}

\subsection{Partially-Dynamic Single-Source Shortest Path Trees}

\begin{definition}
  Let $G=(V,E)$ be an unweighted digraph and let $s\in V$.
  Let $d>0$ be an integer.
  We call an out-tree $T\subseteq G$ rooted in $s$
	a shortest path tree from $s$ up to depth~$d$~if:
  \begin{enumerate}
    \item for any $v\in V$, $v\in V(T)$ iff $\dist_G(s,v)\leq d$, and
    \item for any $v\in V(T)$, $\dist_T(s,v)=\dist_G(s,v)$.
  \end{enumerate}
\end{definition}

\begin{theorem}[Even-Shiloach tree \cite{EvenS81, HenzingerK95}]\label{thm:estree}
  Let $G=(V,E)$ be an unweighted graph subject to partially dynamic edge updates.
  Let $s\in V$ and let $d\geq 1$ be an integer.
  Then, a shortest path tree from $s$ up to depth~$d$ can be explicitly maintained\footnote{By
  this we mean that the algorithm outputs all changes to the edge set
  of the maintained tree.}
  in $O(md)$ total time.
\end{theorem}

\begin{theorem}[$h$-SSSP \cite{Bernstein16}]\label{thm:hsssp}
  Let $G=(V,E)$ be a weighted digraph.
  Let $s\in V$ and let $h\geq 1$ be an integer.
  There exists a partially dynamic algorithm explicitly maintaining
  $(1+\eps)$-approximate distance estimates $\dist'(s,v)$ satisfying
  $$\dist_G(s,v)\leq \dist'(s,v)\leq (1+\eps)\dist_G^h(s,v)$$
  for all $v\in V$.
  The total update time of the algorithm
  $O(mh\log{n}\log(nW)/\eps+\Delta)$.
\end{theorem}
\subsection{Hubs and How to Compute Them}

We now define a blocker set, slightly modifying a definition by King~\cite{King99}.

\begin{definition}\label{def:blocker}
  Let $V$ be a vertex set and let $d$ be a positive integer.
  Let $B\subseteq V$ and let $T$ be a rooted tree over $V$ of depth no more than $d$.
  We call $B$ a
  $(T,d)$-blocker set if for each $v\in V(T)$ such that $\treedep{T}{v}=d$,
  either $v$ or one of its ancestors in $T$ belongs to $B$.
  
  Equivalently, $B$ is a $(T,d)$-blocker set if the tree $T'$ obtained from $T$
  by removing all subtrees rooted in vertices of $B$ (possibly the entire $T$, if
  the root is in $B$) has depth less than $d$.

  Let $\treecol$ be a collection of rooted trees over $V$ of depth no more than $d$.
  We call $B$ a $(\treecol,d)$-blocker set if $B$ is a $(T,d)$-blocker set for
  each $T\in\treecol$.

\end{definition}

\begin{restatable}[\cite{King99}]{lemma}{king}\label{lem:king}
  Let $V$ be a vertex set of size $n$.
Let $d$ be a positive integer. 
  Let $\treecol$ be~a~collection of rooted trees over $V$ of depth at most $d$.
  Then,
  a $(\treecol,d)$-blocker set of size $O\left(\frac{n}{d}\log{n}\right)$
  can be computed deterministically in $O(n\cdot (|\treecol|+ n)\log{n})$ time.
\end{restatable}
\begin{proof}
  For each vertex $v\in V$, we maintain a \emph{score}, defined
  as the sum, over all trees of $T\in \treecol$ that $v$ participates in,
  of the number of descendants $w$ of $v$ such that $\treedep{T}{w}=d$.
  Clearly, the scores can be initialized in $O(|\treecol|\cdot n)$ time.

  The set $B$ is constructed greedily by repeatedly picking a vertex $v$ with
  maximum score, removing from each $T\in \treecol$ the subtree rooted
  at $v$ (including $v$), and updating the scores accordingly.
  To update the scores one needs to iterate through all vertices
  of the removed subtree, as well as all $O(d)$ ancestors of $v$ in $T$.
  Thus, the time to spent on picking maximum-scoring vertices and updating
  the scores can be bounded by $O(|B|\cdot (n+|\treecol|\cdot d)+|\treecol|\cdot n)$.

  As the maximum-scoring vertex participates in at least $\frac{d}{n}$-fraction of the remaining
  $d$-edge root-leaf paths at each step, $O(\frac{n}{d}\log{n})$ vertices
  will end up being picked.
  Hence, the total running time is $O(n\cdot (|\treecol|+ n)\log{n})$.
\end{proof}

\begin{definition}\label{def:covered}
  Let $G=(V,E)$ be a directed graph. Let $B\subseteq V$ and let $d>0$ be an integer.
  We~say that a path $P$ in $G$ is \emph{$(B,d)$-covered} if it can be expressed as
	$P=P_1\ldots P_k$, where $P_i=u_i\to v_i$, $\hops{P_i}\leq d$ for each $i=1,\ldots,k$, and  $u_i\in B$ for each $i=2,\ldots,k$.
\end{definition}
\begin{restatable}{lemma}{coverconcat}\label{lem:cover-concat}
  Let $G=(V,E)$ be a digraph and let $B\subseteq V$.
  Let $P=P_1\ldots P_l$ be a $u\to v$ path and suppose that
  for any $i$, $P_i$ is $(B,d_i)$-covered.
  Then, $P$ is $(B,D)$-covered, where
  $$D=\max\left\{d_x+\ldots+d_z:1\leq x\leq z\leq l\text{ and } V(P_y)\cap B=\emptyset\text{ for all }y\in (x,z)\right\}.\footnote{In this paper we sometimes use the notation $(x,z)$ when referring to the set of integers $y$ satisfying $x<y<z$.}$$
\end{restatable}
\begin{proof}
  Since $P_i$ is $(B,d_i)$-covered, it can be expressed as $P_i=P_{i,1}\ldots P_{i,k_i}$,
  where for each $j\in [1,k_i]$, $\hops{P_{i,j}}\leq d_i$, no vertices of $P_{i,j}$ other than its
  endpoints belong to $B$, and for each $j\in [2,k_i]$, $P_{i,j}$ starts
  in a vertex of $B$.
  Hence, if $V(P_i)\cap B=\emptyset$, then $k_i=1$.
  
  Let $z_1,\ldots,z_r$ be those indices $i$ for which $V(P_i)\cap B\neq\emptyset$ holds,
  in ascending order. Assume that $k_{z_i}\geq 2$ (if $k_{z_i}=1$, we can set $P_{z_i,2}=v_{z_i}\to v_{z_i}$
  to be a path of length $0$).
  Set $P_{l+1,1}$ to be a $0$-edge $v\to v$ path and let $z_{r+1}=l+1$.
  Observe that we can rewrite $D$ as
  
  $$D=\max\left(d_1+\ldots+d_{z_1},d_{z_1}+\ldots+d_{z_2},\ldots,d_{z_r}+\ldots+d_{z_{r+1}}\right).$$
  
  For $i=1,\ldots,r$, let
  $$P_i'=(P_{z_i,2})\ldots (P_{z_i,k_{z_i}}P_{z_i+1,1}P_{z_i+2,1}\ldots P_{z_{i+1},1})$$
  Let $P=(P_{1,1}\ldots P_{z_1,1})P_1'\ldots P_r'$.
  Moreover, all the bracketed subpaths, except of $(P_{1,1}\ldots P_{z_1,1})$, start with a vertex of $B$.
  The maximum number of edges, over all the bracketed subpaths, is clearly no
  more than $D$.
  Therefore, we conclude that $P$ is $(B,D)$-covered.
\end{proof}

\begin{definition}\label{def:hubs}
  Let $G=(V,E)$ be a directed graph and let $d>0$ be an integer. A set $H_d\subseteq V$ is
  called a \emph{$d$-hub set} of $G$ if for every $u,v\in V$ such that $\dist_G(u,v)<\infty$,
  there exists some shortest
  $u\to v$ path that is $(H_d,d)$-covered.
\end{definition}
\begin{restatable}{lemma}{treehubs}\label{lem:tree-hubs}
  Let $G=(V,E)$ be a directed \emph{unweighted} graph and let $d>0$ be an integer.
  Suppose we are given a collection
  $\treecol=\{T_v:v\in V\}$ of shortest
  path trees up to depth~$d$ from all vertices of~$G$.
  Let $B$ be a $(\treecol,d)$-blocker set. Then $B$ is a $2d$-hub set of $G$.
\end{restatable}
\begin{proof}
  Let $u,v\in V$ be any vertices such that $\dist_G(u,v)<\infty$.
  Let $P$ be some shortest $u\to v$ path.
  Let $P_1,\ldots,P_k$, where $k\geq 1$ and $P_i=u_i\to v_i$, be such that $P=P_1\ldots P_k$,
  $\hops{P_1}=\ldots=\hops{P_{k-1}}=d$ and $\hops{P_k}\leq d$.
  Then each $P_i$ is clearly a shortest $u_i\to v_i$ path.
  Since $T_{u_i}$ is a shortest path tree from $u_i$ containing
  all vertices $w$ such that $\dist_G(u_i,w)\leq d$,
  $\dist_{T_{u_i}}(u_i,v_i)=\hops{P_i}$.
  Let $P_i'=\treepath{T_{u_i}}{v_i}$
  and set $P'=P_1'\ldots P_k'$.
  Clearly, $\hops{P'}=\hops{P}$ and thus $P'$ is a shortest $u\to v$
  path.
  As $\hops{P_i'}\leq d$ for each $i$, each $P_i'$ is clearly $(B,d)$-covered.
  All paths $P_1',\ldots,P_{k-1}'$ contain a vertex of $B$ and thus,
  by Lemma~\ref{lem:cover-concat}, $P'$ is $(B,2d)$-covered.
  We conclude that $B$ is indeed a $2d$-hub set of $G$.
  \end{proof}

\subsection{Deterministic Incremental Algorithm for Dense Graphs.}
\begin{theorem}[\cite{King99}+\cite{Bernstein16}]\label{thm:dense-incremental}
There exist an incremental algorithm maintaining $(1+\eps)$-approximate all-pairs distance estimates of a digraph
  in $O(n^3\log^3{n}\log(nW)/\eps+\Delta)$ total time.
\end{theorem}
As mentioned before, the above theorem basically follows by combining the partially-dynamic transitive closure algorithm of King~\cite{King99} with Bernstein's $h$-SSSP
algorithm (Theorem~\ref{thm:hsssp}) for $h=2$.
However, since the algorithm is not stated explicitly anywhere in the literature,
we sketch it in Appendix~\ref{sec:simple-incremental} for completeness.

\section{Deterministic Incremental Algorithm for Sparse Graphs}\label{sec:sparse-incremental}
In this section we present our deterministic incremental algorithm with $\Ot(mn^{4/3}/\eps)$ total update time.
We first observe that whenever an edge $xy$ is added, the set of hubs may be ``fixed" by extending it with both $x$ and $y$.

\begin{restatable}{lemma}{incrementalhubs}\label{lem:incremental-hubs}
  Let $G=(V,E)$ be a directed unweighted graph. Let $H_d$ be a $d$-hub set of $G$.
  Let $x,y\in V$ be such that $xy\notin E$.
  Then $H_d'=H_d\cup \{x,y\}$ is a $d$-hub set of $G'=G+xy$.
\end{restatable}

In the proof we use Lemma~\ref{lem:cover-concat-simple}, whose proof can be found above.
\begin{restatable}{lemma}{coverconcatsimple}\label{lem:cover-concat-simple}
  Let $G=(V,E)$ be a directed graph and let $B\subseteq V$.
  Let $b\in B$.
  Let $P$ be a $u\to b$ path in $G$ that is $(B,d_P)$-covered.
  Let $Q$ be a $b\to v$ path in $G$ that is $(B,d_Q)$-covered.
  Then the path $PQ=u\to v$ is $(B,\max(d_P,d_Q))$-covered.
\end{restatable}
\begin{proof}
  Let $P=P_1\ldots P_k$ be such that for all $i$, $\hops{P_i}\leq d_P$
  and for all $i\geq 2$, $P_i$ starts at a vertex of $B$.
  Similarly, 
  let $Q=Q_1\ldots Q_l$ be such that for all $i$, $\hops{Q_i}\leq d_Q$
  and for all $i\geq 2$, $Q_i$ starts at a vertex of $B$.
  By $b\in B$, all paths $P_2,\ldots,P_k,Q_1,\ldots,Q_l$
  start with a vertex of $B$.
  Moreover, each path of $P_1,\ldots,P_k,Q_1,\ldots,Q_l$ has
  length no more than $\max(d_P,d_Q)$.
  We conclude that the path $PQ$ is $(B,\max(d_P,d_Q))$-covered.
\end{proof}

\begin{proof}[Proof of Lemma~\ref{lem:incremental-hubs}]
  Let $u,v\in V$ be such that $\dist_{G'}(u,v)<\infty$.
  We need to show that there exists a shortest path $P=u\to v$ in $G'$ such that
  $P$ is $(H_d',d)$-covered.

  Since $G \subseteq G'$, $\dist_{G'}(u,v) \leq \dist_G(u,v)$
  If $\dist_{G'}(u,v)=\dist_G(u,v)$, then there exists a shortest path $P=u\to v$
  such that $P\subseteq G\subseteq G'$.
  Hence, $P$ is $(H_d,d)$-covered and thus also $(H_d',d)$-covered.
  
  If $\dist_{G'}(u,v)<\dist_G(u,v)$, then all shortest $u\to v$ paths
  go through the edge $xy$.
  Let $P=P_1(xy)P_2$ be any such path, where $P_1=u\to x$ and $P_2=v\to y$
  are shortest paths in $G$.
  Hence, there also exist shortest paths $P_1'=u\to x$ and $P_2'=v\to y$,
  also contained in $G$, such that $P_1'$ and $P_2'$ are both
  $(H_d,d)$-covered and hence also $(H_d',d)$-covered.
  Observe that the single-edge path $(x,y)$ is $(H_d',1)$-covered.

  By Lemma~\ref{lem:cover-concat-simple},
  since $x\in H_d'$ and $d\geq 1$, the path $P_1'(x,y)$ is $(H_d',d)$ covered.
  Again by Lemma~\ref{lem:cover-concat-simple}, since $y\in H_d'$,
  and both paths $P_1'(x,y)$ and $P_2'$ are $(H_d',d)$-covered, $P_1'(x,y)P_2'$ is 
  $(H_d',d)$-covered as well.
  We conclude that $H_d'$ is indeed a $d$-hub set of $G'$.
\end{proof}
\subsection{The Data Structure Components}

Let $d>1$ be an even integer and let $\eps_1$, $0<\eps_1<\eps$ be a real number, both to be set later.
Our data structure consists of several components. Each subsequent
component builds upon the previously defined components only.
\paragraph{Exact shortest paths between nearby vertices.}
The data structure maintains two collections $\treecol^\tfrom=\{T^\tfrom_v: v\in V\}$ and $\treecol^\tto=\{T^\tto_v : v\in V\}$
of shortest path trees up to depth $\frac{d}{2}$ 
in $G$ and $\rev{G}$, resp.
By Theorem~\ref{thm:estree}, each tree of $\treecol^\tfrom\cup\treecol^\tto$ can be maintained under edge insertions
in $O(md)$ total time. The total time spent in this component is hence $O(nmd)$.
\paragraph{The hubs.}
A $d$-hub set $H_d$ of both $G$ and $\rev{G}$ such that
$|H_d|=O\left(\frac{n}{d}\log{n}\right)$ is maintained at all times, as follows.
Initially, $H_d$ is computed in $O(n^2\log{n})$ time using Lemma~\ref{lem:king}
and the trees of $\treecol^\tfrom\cup \treecol^\tto$ (see Lemma~\ref{lem:tree-hubs}).
Next, the data structure operates in phases.
Each phase spans $f=\Theta(\frac{n}{d}\log{n})$ consecutive edge insertions.
When an edge $xy$ is inserted, its endpoints are inserted into $H_d$.
By Lemma~\ref{lem:incremental-hubs}, this guarantees that $H_d$ remains
a $d$-hub set of both $G$ and $\rev{G}$ after the edge insertion.
Once $f$ edges are inserted in the current phase, the phase ends
and the hub set $H_d$ is recomputed from scratch, again using Lemma~\ref{lem:king}.
Observe that the size of $|H_d|$ may at most triple within each phase.

The total time spent on maintaining the set $H_d$
is clearly $O\left(\frac{m}{f}\cdot n^2\log{n}\right)=O(nmd)$.
\paragraph{Approximate shortest paths between the hubs.}
In each phase, we maintain a \emph{weighted} graph $A=(H_d,E_A)$, where
$E_A=\{uv:u,v\in H_d, \dist_{\rev{G}}(u,v)\leq d\}$
and $\wei_A(uv)=\dist_{T^\tto_u}(u,v)=\dist_{\rev{G}}(u,v)\leq d$.
Observe that during each phase, the graph $A$ is in fact incremental.
We can thus maintain $(1+\eps_1)$-approximate distance estimates
$\dist_A'(u,v)$ for all $u,v\in H_d$
in $O(|H_d|^3\log^4{n}/\eps_1)=O\left(\left(\frac{n}{d}\right)^3\log^7{n}/\eps_1\right)$
total time per phase, using a data structure $\dstr_A$ of Theorem~\ref{thm:dense-incremental}.\footnote{Technically
speaking, the total update time of the data structure of Theorem~\ref{thm:dense-incremental}
is $O(n^3\log^4{n}/\eps')+O(\Delta)$.
However, all updates to $\dstr_A$ arise when some previous component
updates its explicitly maintained estimates,
so the $\Delta$ term is asymptotically no more than the total update time
of the previously defined components and can be charged to those. In the following,
we omit $\Delta$ terms like this without warnings.}

Summing over all phases, the total time spent on maintaining
the $(1+\eps')$-approximate
distances estimates of the graph $A$ is
$O\left(m\left(\frac{n}{d}\right)^2\log^6{n}/\eps_1\right)$.
\begin{restatable}{lemma}{shortpathapsp}\label{lem:short-path-apsp}
  For any $u,v\in H_d$, $\dist_{\rev{G}}(u,v)=\dist_A(u,v)$.
\end{restatable}
\begin{proof}
  Let $u,v\in H_d$. If $\dist_{\rev{G}}(u,v)=\infty$, then clearly $\dist_A(u,v)=\infty$
  as the edges of $A$ correspond to paths in $\rev{G}$.
  For the same reason, we have $\dist_A(u,v)\geq\dist_{\rev{G}}(u,v)$.

  Suppose $\dist_{\rev{G}}(u,v)<\infty$.
  By the definition of $H_d$, there exists a shortest path $P=u\to v$ in $\rev{G}$
  such that $P$ is $(H_d,d)$-covered.
  Hence, $P$ can be expressed as $P_1\ldots P_k$, where $P_i=u_i\to v_i$,
  $u_i,v_i\in H_d$ and $\hops{P_i}\leq d$ for all $i$.
  Since $P_i$ is a shortest $u_i\to v_i$ path in $\rev{G}$, the edge $u_iv_i$
  has length $\hops{P_i}$ in $A$.
  Consequently, the path $(u_1v_1)(u_2v_2)\ldots (u_kv_k)=u\to v$ has
  length $\hops{P}=\dist_{\rev{G}}(u,v)$ in~$A$.
  We thus obtain $\dist_A(u,v)\leq\dist_{\rev{G}}(u,v)$,
  which implies $\dist_A(u,v)=\dist_{\rev{G}}(u,v)$.
\end{proof}

By the above lemma, for each $u,v\in H_d$ we actually have $\dist'_A(u,v)\leq (1+\eps_1)\dist_{\rev{G}}(u,v)$.
\paragraph{Shortcuts to hubs.}
For each $u\in V$, let $S_u$ be a graph on $V$ with exactly $n$
edges $\{uv:v\in V\}$ satisfying $\wei_{S_u}(u,v)\geq \dist_{\rev{G}}(u,v)$
for all $v\in V$, and additionally $\wei_{S_u}(u,v)\leq (1+\eps_1)\dist_{\rev{G}}(u,v)$ if $u,v$ both 
currently belong to $H_d$.
The edges between vertices of $H_d$ are the only ones that our algorithm needs to compute approximate distances.
For other edges we only need to make sure they will not cause the algorithm to underestimate the distances.

Observe that the graphs $S_u$ can be maintained using the previously
defined components as follows. First, they are initialized so that their
edges are all infinite-weight.
Whenever the data structure $\dstr_A$ changes
(or initializes) some of its estimates $\dist_A'(u,v)\leq (1+\eps_1)\dist_{\rev{G}}(u,v)$,
we perform $\wei_{S_u}(u,v):=\min(\wei_{S_u}(u,v),\dist_A'(u,v))$.
This guarantees that the invariants posed on $S_u$ are always
satisfied and $S_u$ is incremental.
The total number of updates to all graphs $S_u$ is equal to the number
of estimate updates made by $\dstr_A$ and thus can be neglected.

For $u \in V$, we set up a $h$-SSSP data structure $\dstr_{u}$ of Theorem~\ref{thm:hsssp}
for the graph $\rev{G}\cup S_u$ with source vertex $u$ and  $h=d+1$.
Hence, $\dstr_{u}$ maintains distance estimates $\dist'(u,v)$ such that
$\dist'(u,v)\leq (1+\eps')\dist_{\rev{G}\cup S_u}^{d+1}(u,v)$.
As the graph $\rev{G}\cup S_u$ is incremental and has $O(m)$ edges, the total
time that $\dstr_u$ needs to operate is
$O(md\log^2{n}/\eps_1+\Delta_u)$,
where $\Delta_u$ is the total number of updates 
to $\rev{G}\cup S_u$.
Summing the update times for all data structures $\dstr_u$, we obtain
$O(nmd\log^2{n}/\eps_1+\sum_{v\in V} \Delta_u)$
total
time.
Note that $\sum_{u\in V}\Delta_u$ equals $nm$ plus the number
of updates to the graphs $S_u$, which can be charged to the operating
cost of data structure $\dstr_A$, as argued before.
We conclude that the total update time of all $\dstr_u$ is
$O(nmd\log^2{n}/\eps_1)$.

Observe that a shortest $u \to v$ path in $\rev{G}$, where $u \in H_d$ and $v \in V$ is approximated by a path in $\rev{G} \cup S_u$ consisting of at most $d+1$ edges.
The first edge belongs to $S_u$ and ``jumps" to some hub.
The latter (at most $d$) edges belong to $\rev{G}$.
This is formalized as follows.

\begin{restatable}{lemma}{incrshortcutrev}\label{lem:incr-shortcut-rev}
  Let $u\in H_d$ and $v\in V$. Then $\dist_{\rev{G}\cup S_u}^{d+1}(u,v)\leq (1+\eps_1)\dist_{\rev{G}}(u,v)$.
\end{restatable}
\begin{proof}
  Recall that $H_d$ is a $d$-hub set of $\rev{G}$.
  It follows that there exists a shortest path $P=P_1\ldots P_k=u\to v$
  such that for each $i$, $P_i=u_i\to v_i$, $\hops{P_i}\leq d$ and $u_i\in H_d$
  (since $u_1=u\in H_d$).
  Consider the edge $e=uu_k$ of $S_u$.
  As $u,u_k\in H_d$, $\wei_{S_u}(e)\leq (1+\eps_1)\dist_{\rev{G}}(u,u_k)=
  (1+\eps_1)\len(P_1\ldots P_{k-1})$.
  Moreover, clearly $P_k\subseteq \rev{G}\cup S_u$.
  The path $P'=eP_k=u\to v$ thus has $d+1$ edges and $P'\subseteq \rev{G}\cup S_u$.
  Consequently, we obtain
  \begin{align*}
    \dist_{\rev{G}\cup S_u}^{d+1}(u,v) &\leq\len(P')\\
    &=\wei_{S_u}(e)+\len(P_k)\\
    &\leq(1+\eps_1)\len(P_1\ldots P_{k-1})+\len(P_k)\\
    &\leq (1+\eps_1)\len(P)\\
    &=(1+\eps_1)\dist_{\rev{G}}(u,v).
    \qedhere
  \end{align*}
\end{proof}

By the above lemma, we conclude that for $u\in H_d$, $v\in V$, 
the estimate $\dist'(u,v)$ produced by the data structure $\dstr_v$
satisfies $\dist'(u,v)\leq (1+\eps_1)\dist_{\rev{G}\cup S_u}^{d+1}(u,v)\leq (1+\eps_1)^2\dist_{\rev{G}}(u,v)$.

\paragraph{All-pairs approximate shortest paths.} We maintain
another set of shortcut graphs $R_u$, for $u\in V$.
Again $R_u$ has exactly $n$ edges $\{uv:v\in V\}$
whose weights satisfy $\wei_{R_u}(uv)\geq \dist_G(u,v)$
for all $v$ and $\wei_{R_u}(uv)\leq (1+\eps_1)^2 \dist_G(u,v)$
if $v\in H_d$.
Each graph $R_u$ is maintained using the previously defined
data structures $\dstr_v$.
Initially all weights of $R_u$ are infinite.
Whenever some $\dstr_v$ changes the estimate
$\dist'(v,u)$, we set $\wei_{R_u}(uv):=\min(\wei_{R_u}(uv),\dist'(v,u))$.
Since for $v\in H_d$ we have $\dist'(v,u)\leq (1+\eps_1)^2\dist_{\rev{G}}(v,u)$,
equivalently, $\dist'(v,u)\leq (1+\eps_1)^2\dist_{G}(u,v)$
and we obtain $\wei_{R_u}(uv)\leq (1+\eps_1)^2\dist_G(u,v)$.
Therefore, the graphs $R_u$ are all incremental and the total number
of changes they are subject to is no more than the total
number of estimate changes made by the data structures
$\dstr_v$, $v\in V$.
Thus, we may neglect the cost of actually performing these changes.

Finally, for each $u\in V$ we set up a $h$-SSSP data structure $\dstr'_u$ of Theorem~\ref{thm:hsssp}
on graph $G\cup R_u$ with source $u$ and $h=d+1$,
maintaining $(1+\eps_1)$-approximate estimates of the values $\dist^{d+1}_{G\cup R_u}(u,\cdot)$.
Similarly as was the case for the data structures $\dstr_u$ of the previous
component, as the graphs $G\cup R_u$ are incremental,
the total operating time of the $h$-SSSP instances
running on the graphs $G\cup R_u$ is
$O(nmd\log^2{n}/\eps_1)$.

\begin{restatable}{lemma}{incrshortcut}
  Let $u, v\in V$. Then $\dist_{G\cup R_u}^{d+1}(u,v)\leq (1+\eps_1)^2 \dist_G(u,v)$.
\end{restatable}
\begin{proof}
  Recall that $H_d$ is a $d$-hub set of $G$. Hence, if $\dist_G(u,v)<\infty$, there
  exists a shortest path $P=P_1\ldots P_k=u\to v$, where $P_i=u_i\to v_i$,
  $\hops{P_i}\leq d$ and $u_i\in H_d$ for all $i\geq 2$.
  If $k=1$ then $\hops{P}\leq d$ and as a result $\dist_{G\cup R_u}^{d+1}(u,v)\leq \dist_G^d(u,v)=\dist_G(u,v)$.

  Suppose $k\geq 2$. Then $u_k\in H_d$ and consequently the weight of the edge $e=uu_k\in E(R_u)$
  satisfies $\wei_{R_u}(e)\leq (1+\eps_1)^2\dist_G(u,u_k)=(1+\eps_1)^2\hops{P_1\ldots P_{k-1}}$.
  The path $P'=eP_k\subseteq G\cup R_u$ satisfies $\len(P')\leq (1+\eps_1)^2\len(P)=(1+\eps_1)^2\dist_G(u,v)$
  and $\hops{P'}\leq d+1$.
  Hence, indeed $\dist_{G\cup R_u}^{d+1}(u,v)\leq\len(P')\leq (1+\eps_1)^2\dist_G(u,v)$.
\end{proof}

By the above lemma, the the distance estimates $\dist''(u,v)$ maintained
by the data structure $\dstr'_u$, approximate the corresponding
distances $\dist_G(u,v)$ within a factor of $(1+\eps_1)^3$.

\begin{restatable}{theorem}{incrtheorem}
  Let $G$ be a directed unweighted graph. There exists a deterministic incremental
  algorithm maintaining $(1+\eps)$-approximate distance estimates between all
  pairs of vertices of $G$ in
  $O(mn^{4/3}\log^{10/3}{n}/\eps)$
  total time.
\end{restatable}
\begin{proof}
  Let $\eps_1=\eps/6$. Then the estimates $\dist''(u,v)$ produced by the data
  structures $\dstr'_u$ approximate the actual distances
  within a factor of $\left(1+\frac{\eps}{6}\right)^3\leq 1+\eps$.
  The total time used by the data structure is clearly
  $O\left(nmd\log^2{n}/\eps+m\left(\frac{n}{d}\right)^2\log^6{n}/\eps\right)=O\left(nm\log^2{n} \cdot \left(d+\frac{n}{d^2}\log^4{n}\right)/\eps\right)$.
  By setting
  $d=n^{1/3}\log^{4/3}{n}$,
  we obtain the desired bound.
\end{proof}

\section{Partially-Dynamic Verification of a Sampled Hub Set}\label{sec:verification}

In this section we show how to maintain the information whether a sampled set
remains a hub set of an unweighted digraph $G$ subject to partially dynamic updates.
For simplicity, assume that $G$ is decremental (the incremental case, being somewhat
easier, can be handled similarly).
We start by showing how a reliable hub set can be found if we are given
shortest path trees up to depth $d$ from all vertices of $G$,
stored in dynamic tree data structures.


\begin{restatable}[\cite{UY91,Zwick02}]{lemma}{randomsubset}\label{lem:random-subset}
  Let $V$ be a vertex set of size $n$ and let $d>0$ be an integer.
  Let $\treecol$ be a collection of rooted trees of depth no more than $d$,
  whose vertex sets are subsets of $V$.

  Let $c>1$ be some positive constant.
  Let $B$ be a random subset of $V$ of size $\min\left(\lceil c\frac{n}{d}\ln{n}\rceil, n\right)$.
  Then, $B$ is a $(\treecol,d)$-blocker set with probability at least
  $\max(0,1-|\treecol|/n^{c-1})$.
\end{restatable}
\begin{proof}
  If $\lceil c\frac{n}{d}\ln{n}\rceil\geq n$, then $B=V$ and the lemma obviously holds.
  
  Suppose $\lceil c\frac{n}{d}\ln{n}\rceil<n$.
  For each $v\in V$, $\prob[v\notin B]\leq 1-\frac{c\ln{n}}{d}$.
  Hence, for a fixed $T\in\treecol$ and $w\in T$ such that $\treedep{T}{w}=d$,
  the probability that neither $w$ nor any of its ancestors in $T$ belong
  to $B$ is no more than $\left(1-\frac{c\ln{n}}{d}\right)^{d+1}\leq e^{-c\ln{n}}=n^{-c}.$
  Therefore, the probability that this happens for any $T$ and $w$,
  is no more than $|\treecol|/n^{c-1}$.
\end{proof}


\begin{restatable}{lemma}{dyntree}\label{lem:dyntree}
  Let $V$ be some set of $n$ vertices. Let $F$ be a forest of (initially single-vertex) rooted
  out-trees over $V$ such that the vertex sets of the individual trees of $F$
  form a partition of $V$.
  For $v\in V$, let $T_v\in F$ denote the unique tree of $F$ containing $v$.

  There exists a data structure for
  dynamically maintaining $F$ (initially consisting of $n$ $1$-vertex trees)
  and supporting the following operations in $O(\log{n})$ time each:
  \begin{enumerate}
    \item $\tparop(v)$: if $v$ is not the root of $T_v$, return its parent. Otherwise return $\nil$.
    \item $\tlink(u,v)$: assuming $T_u\neq T_v$ and that $u$ is the root of $T_u$, make $u$ a child
      of $v$ in $T_v$.
    \item $\tcut(v)$: assuming $v$ is not the root of $T_v$, split $T_v$ into two trees
      by removing the edge between $v$ and its parent.
    \item $\tdepth(v)$: return the depth of the tree $T_v$.
  \end{enumerate}

\end{restatable}
\begin{proof}
  We store our forest $F$ as a collection of the Euler-tour based dynamic trees of Tarjan \cite{Tarjan97}.
  This data structure is capable of maintaining a forest $F'$ of \emph{unrooted} and \emph{undirected}
  trees over $V$ such that each $v\in V$ is additionally associated
  a real value $val(v)$. 
  Similarly, let us denote by $T_v'$ the unique tree of the forest containing $v$.
  The supported operations are as follows:
  \begin{itemize}
    \item $link(u,v)$: if $T_u'\neq T_v'$, then combine $T_u'$ and $T_v'$
      by adding the edge $uv$.
    \item $cut(u,v)$: assuming $T_u'=T_v'$ and $uv\in E(T_u')$, break
      the tree $T_u'$ into two by removing the edge $uv$.
    \item $change\textnormal{-}val(v,x)$: set $val(v):=x$.
    \item $find\textnormal{-}val(v)$: return $val(v)$.
    \item $add\textnormal{-}val(v,x)$: add $x\in \mathbb{R}$ to $val(w)$ of every vertex $w\in V(T_v')$.
    \item $find\textnormal{-}max\textnormal{-}val(v)$: return a vertex $w$ with maximum $val(w)$ in $T_v'$.
  \end{itemize}

  Even though the data structure of \cite{Tarjan97} maintains a forest
  of unrooted undirected trees, we show that the above operations
  are sufficient to maintain the depth of each out-tree (which both rooted and directed) of $F$.

  To be able to operate on out-trees, for each vertex $v\in V$, we explicitly store its parent $\tpar(v)$
  in a table (if $v$ is the root
  of $T_v$, $\tpar(v)=\nil$).
  Initially we set $\tpar(v)=\nil$ for all $v\in V$.
  $\tparop(v)$ simply returns $\tpar(v)$.
  The forest $F'$ stored in our data structure of \cite{Tarjan97} always consists of (undirected) edges of the form $\tpar(v)v$,
  for $v$ such that $\tpar(v)\neq\nil$.
  Hence, $F'$ is initially empty.

  For each $v$, we will guarantee that the value $val(v)$ maintained by the data structure of \cite{Tarjan97} is
  equal to $\treedep{T_v}{v}$.
  To this end, we initially perform $change\textnormal{-}val(v,0)$ for each $v\in V$.
  By the invatiant posed on the values $val(v)$, and since the depth of any $T_v$ is the maximum depth over all its vertices,
  we can implement $\tdepth(v)$ with $find\textnormal{-}max\textnormal{-}val(v)$.

  $\tcut(v)$ is implemented as follows. We first record $y:=find\textnormal{-}val(v)$.
  Then we perform $cut(\tpar(v),v)$ on $F'$ and set $\tpar(v)$ to $\nil$.
  Afterwards, $T_v$ is a tree rooted in $v$.
  At this point the values $val(w)$ for all $w\in V(T_v)$ exceed the
  respective depths in $T_v$ by exactly $y$.
  So, in order to update the vertex depths in the new tree $T_v$, we perform
  $add\textnormal{-}val(v,-y)$.
  Clearly, the depths in the other obtained tree (the one that does not contain $v$) need not be updated.

  To implement $\tlink(u,v)$, we first set $\tpar(u):=v$. Then, we update
  the depths in $T_u$ before it is connected to $T_v$ by calling
  $add\textnormal{-}val(u,find\textnormal{-}val(v)+1)$.
  Finally, we call $link(u,v)$.

  Since each of the operations $\tdepth$, $\tcut$ and $\tlink$ translate
  into a constant number of operations on the data structure of \cite{Tarjan97} storing $F'$,
  all operations run in $O(\log{n})$ time.
\end{proof}

\begin{lemma}\label{lem:las-vegas-static}
  Let $V$ be a vertex set, $n=|V|$, and let $d>0$ be integral.
  Let $\treecol$ be a collection of rooted trees over $V$ of depth no more than $d$,
  where $|\treecol|=O(\poly{n})$.
  Suppose each $T\in\treecol$ is given as a separate data structure
  of Lemma~\ref{lem:dyntree} and for each $T\in\treecol$, $\troot(T)$ is known.

  Then, there exists a Las Vegas randomized algorithm computing
  a $(\treecol,d)$-blocker set $B$ of size $O\left(\frac{n}{d}\log{n}\right)$
	in $O(|\treecol|\cdot \frac nd \cdot \log^2{n})$ time
  with high probability.
\end{lemma}
\begin{proof}
  Let $|\treecol|=O(n^\alpha)$ for some $\alpha>0$.
  The algorithm is to simply repeatedly pick random subsets $B$ of $V$
  of size $\min(\lceil (\alpha+2)\frac{n}{d}\ln{n}\rceil,n)$
  until $B$ succeeds in being a $(\treecol,d)$-blocker set.
  By Lemma~\ref{lem:random-subset}, for a random $B$, the probability
  that this is not the case is at most $\frac{1}{n}$.
  Hence, the probability that we fail finding a $(\treecol,d)$-blocker 
  set after $k=O(1)$ trials is at most $1/n^k$.

  We thus only need to show how to verify whether a set $B$
  is actually a $(\treecol,d)$-blocker set in $O(|\treecol|\cdot |B|\log{n})$ time.
  Recall that for a single $T\in\treecol$, if the depth of $T$ is no
  more than $d$, 
  then~$B$ is a $(T,d)$-blocker set if the tree $T'$ obtained from $T$
  by removing all subtrees rooted in vertices of~$B$, has depth less than $d$.
    Consequently, to verify whether $B$ is a $(T,d)$-blocker set, we take
  advantage of the fact that $T$ is stored in a data structure
  of Lemma~\ref{lem:dyntree}.

  We first check whether $r=\troot(T)\in B$.
  If this is the case, $B$ is a $(T,d)$-blocker set in a trivial way.
  Otherwise, for each $b\in B$, we store $p_b=\tparop(b)$ and perform $\tcut(b)$.
  Afterwards, one can see that $B$ is a $(T,d)$-blocker set if and only if $\tdepth(r)<d$.
  Finally, we revert all the performed $\tcut$ operations by running $\tlink(b,p_b)$
  for all $b\in B$.

  Clearly, the time needed to verify whether $B$ is a $(T,d)$-blocker set
  for any $T\in\treecol$, is $O(|B|\log{n})$.
  Hence, one can check whether $B$ is a $(\treecol,d)$-blocker set in $O(|\treecol|\cdot|B|\log{n})$
  time.
\end{proof}

Now we move on to the problem of detecting when a sampled set
ceases to be a valid hub set of $G$. 
In fact, our algorithm will solve a bit more general problem (which is anyway needed
for applications, as we will see later), as follows.

Let $|V|=n=a_0> a_1> \ldots > a_q=1$ be some sequence of integers
such that $a_i\mid a_{i-1}$.
For each $i=0,\ldots,q$, let $A_i$ be a random $a_i$-subset (a subset of size $a_i$)
of $V$.
By Lemmas~\ref{lem:tree-hubs}~and~\ref{lem:random-subset}, each $A_i$ is in fact 
an $\Theta\left((n / a_i)\ln{n}\right)$-hub set of $G$
with high probability.

We would like to detect when some $A_i$ ceases to be
an $\Theta\left((n / a_i)\ln{n}\right)$-hub set of $G$
while $G$ undergoes edge deletions.
Using this terminology, both state-of-the-art
Monte-Carlo randomized algorithms for decremental exact shortest paths \cite{BaswanaHS07}
and partially-dynamic $(1+\eps)$-approximate shortest paths \cite{Bernstein16} (for unweighted digraphs) use randomness
only for constructing hub sets $A_0,\ldots,A_q$ (they use $a_i=2^{q-i}$, but in fact any $a_i=c^{q-i}$, where $c$ is a positive integer,
would be sufficient for these algorithms to work),
valid simultaneously for all versions of the input graph with high probability
(the sets $A_i$ satisfy this, as we will later show).

Without loss of generality, we can assume that given the sets $A_0,\ldots,A_q$,
the algorithms of \cite{BaswanaHS07, Bernstein16} proceed deterministically
(we discuss these algorithms in a more detailed way in Appendix~\ref{s:baswanabernstein}).
Suppose we develop an efficient partially dynamic algorithm $\mathcal{A}$
\emph{verifying} whether each $A_i$ remains a
$\Theta\left((n / a_i)\ln{n}\right)$-hub set of $G$ (i.e.,
$\mathcal{A}$ is supposed to detect that some $A_i$ ceases
to be a $\Theta\left((n/a_i)\ln{n}\right)$-hub set immediately
after this happens)
and producing \emph{false negatives} with low probability
(the algorithm is guaranteed to be correct if it says that all
$A_i$ have the desired property but might be wrong saying that some $A_i$ is no longer
a hub set).
Then, we could use $\mathcal{A}$ to convert the algorithms
of \cite{BaswanaHS07, Bernstein16} into \emph{Las Vegas} algorithms
by drawing new sets $A_0,\ldots,A_q$ and
restarting the respective algorithms whenever $\mathcal{A}$ detects (possibly incorrectly)
that any of these sets ceases to be a hub set.
As this does not happen w.h.p.,
with high probability the overall asymptotic running time remains unchanged.
The remainder of this section is devoted to describing such an algorithm $\mathcal{A}$.

\begin{restatable}{lemma}{treemaintain}\label{lem:tree-maintain}
  Let $d>0$ be an integer.
  Let $F$ be a forest of out-trees of depth no more than $d$ over~$V$.
  Denote by $T_v$ the unique tree of $F$ containing $v\in V$.
  Let $B\subseteq V$ be fixed.
  
  There exists a data structure with update time $O(\log{n})$, maintaining the information whether $B$
  is a $(F,d)$-blocker set, subject to updates to $F$ of the following types:
  \begin{itemize}
    \item cut the subtree rooted in $v$ out of $T_v$ where $v\in V$ and $v$ is not the root of $T_v$,
    \item make the tree $T_r$ a child of $v\in T_v$ where $r\in V$ is the root of $T_r$ and $v\notin T_r$,
  \end{itemize}
\end{restatable}
\begin{proof}
    We use the data structure of Lemma~\ref{lem:dyntree} to store $F$.
  However, each $T\in F$ is represented in this data structure
  as a collection $F_T$ of either $|V(T)\cap B|$ (if the root of $T$ is in $B$) or $|V(T)\cap B|+1$ (otherwise) trees:
  the maximal subtrees of $T$ that intersect with $B$ either
  only in their respective roots or not at all.

  Recall that for a single $T\in F$, if the depth of $T$ is no
  more than $d$, 
  then $B$ is a $(T,d)$-blocket set if the tree $T'$ obtained from $T$
  by removing all subtrees rooted in vertices of $B$ has depth less than~$d$.
  Equivalently, $B$ is a $(T,d)$-blocker set if the depth
  of all trees of $F_T$ is less than $d$.
  Consequently, $B$ is a $(F,d)$-blocker set if the depth
  of all trees of $\bigcup_{T\in F} F_T$ is less than $d$.
  This information is easy to maintain using
  the data structure of Lemma~\ref{lem:dyntree} if this
  data structure stores $F'=\bigcup_{T\in F} F_T$
  at all times.
  This is precisely what our algorithm does.

  The first operation can be implemented simply as $\tcut(v)$ on $F'$.
  To implement the second operation, we first check if $r\in B$.
  If so, we do nothing, as the vertices of $B$ can only be roots in~$F'$.
  Otherwise, we perform $\tlink(r,v)$.
  After any $\tlink(r,v)$ operation, to update the information whether $B$ is still a $(F,d)$-blocker set, it is sufficient
  to check whether $\tdepth(v)<d$.
\end{proof}

The following technical lemma will prove useful.

\begin{restatable}{lemma}{hubsatleast}\label{lem:hubs-atleast}
  Let $G=(V,E)$ be a directed graph and let $B\subseteq V$.
  Let $d>0$ be an integer.
  Let $P$ be a path that is $(B,d)$-covered and $\hops{P}\geq d$.
  Then, one can represent $P$ as $P_1\ldots P_k$
  such that:
  \begin{itemize}
    \item $\hops{P_i}\in [d,3d]$ for each $i=1,\ldots,k$,
    \item $P_i$ starts with a vertex of $B$ for each $i=2,\ldots,k$.
  \end{itemize}
\end{restatable}
\begin{proof}
  By the fact that $P=u\to v$ is $(B,d)$-covered, $P$ can be expressed as $P=P'_1\ldots P'_l$,
  where $P'_i=u_i\to v_i$,
  $\hops{P'_i}\leq d$ for all $i$ and $u_i\in B$ for all $i\geq 2$.
  Let us partition the path $P'_1\ldots P'_l$ into subpaths (blocks) $P_1,\ldots,P_k$
  using the following procedure.
  
  At each step $P'_j,\ldots,P'_l$ ($j\geq 1$) will constitute the
  remaining part of $P'_1,\ldots,P'_l$ to be partitioned and $\hops{P'_j\ldots P'_l}\geq d$.
  Initially, $j=1$.
  If $\hops{P'_j\ldots P'_l}\leq 3d$, the next block $P_i$
  is set to $P'_j\ldots P'_l$ and the procedure ends.
  Otherwise, let $f$ be
  the minimum index such that $\hops{P'_j\ldots P'_f}\geq d$.
  We set the next block $P_i$ to be $P'_j\ldots P'_f$.
  By the definition of $f$ and $\hops{P'_f}\leq d$, we have
  $\hops{P_i}=\hops{P'_j\ldots P'_f}\leq 2d$ and $\hops{P'_{f+1}\ldots P'_l}\geq d$.
\end{proof}

The following lemma says that in order to test whether a given set of vertices is a $6d$-hub set
it suffices to test the hub set property for paths starting in vertices of a $d$-hub set.

\begin{restatable}{lemma}{hubsinductive}\label{lem:hubs-inductive}
  Let $G=(V,E)$ be a directed unweighted graph. Let $H_d$ be a $d$-hub set of $G$.
  Suppose we are given two collections 
  $\treecol^\tfrom=\{T^\tfrom_v:v\in H_d\}$,   
  $\treecol^\tto=\{T^\tto_v:v\in H_d\}$
  of shortest path trees up to depth $d$ from all vertices of $H_d$ in $G$ and $\rev{G}$,
  respectively.

  Let $B$ be a $(\treecol^\tfrom\cup \treecol^\tto,d)$-blocker set.
  Then $B$ is a $6d$-hub set of $G$.
\end{restatable}
\begin{proof}
  Let $u,v\in V$ be any vertices such that $\dist_G(u,v)<\infty$.
  If $\dist_G(u,v)\leq d$, then any shortest $u\to v$ path
  is $(B,6d)$-covered.
  Suppose $\dist_G(u,v)>d$.

  Let $P$ be any shortest $u\to v$ path in $G$ that is $(H_d,d)$-covered.
  By Lemma~\ref{lem:hubs-atleast}, one can express $P$ as $P=P_1\ldots P_k$ where
  $\hops{P_i}\in [d,3d]$ and each $P_i$ for $i\geq 2$ starts with a vertex of $H_d$.
  Now define $P_1',\ldots,P_k'$ as follows. 
  If $i\in [2,k]$, then
  let $P_i=R_iS_i$ where $R_i=a_i\to b_i$, $\hops{R_i}=d$, and set $P_i'=\treepath{T^\tfrom_{a_i}}{b_i}\cdot S_i.$
  For $i=1$, suppose $P_1=R_1S_1$ where $\hops{S_1}=d$ and $S_1=a_1\to b_1$.
  Then we set $P_1'=R_1\cdot \rev{(\treepath{T^\tto_{b_i}}{a_i})}$.
  Clearly, for all $i$, $P_i'$ has the same endpoints as $P_i$ and $\hops{P_i'}=\hops{P_i}$.
  Therefore, $P'=P_1'\ldots P_k'$ is also a shortest $u\to v$ path.

  Since $\hops{P_i'}\leq 3d$, $P_i'$ is $(B,3d)$-covered.
  Moreover, by the definition of $B$, for all $i$ we have $V(P_i')\cap B\neq\emptyset$.
  Hence, by Lemma~\ref{lem:cover-concat}, we conclude that $P'$ is 
  $(B,6d)$-covered.

  We have thus proved that for all $u,v\in V$, where $\dist_G(u,v)<\infty$,
  some shortest $u\to v$ is $(B,6d)$-covered.
  Equivalently, $B$ is a $6d$-hub set of $G$.
\end{proof}

Observe that by Lemma~\ref{lem:random-subset}, there exists an integral constant $z>0$,
such that for any fixed collection of trees $\treecol$ of depth no more than $z\cdot \frac{n}{a_i}\lceil\ln{n}\rceil$,
where
$|\treecol|=O(n^3)$, $A_i$ is a $\left(\treecol,z\cdot \frac{n}{a_i}\lceil\ln{n}\rceil\right)$-blocker set
with high probability.
For $i=0,\ldots,q$, set $d_i=z\cdot \frac{n}{a_{i+1}}\lceil\ln{n}\rceil$ where $a_{q+1}=1$.
Suppose $G$ undergoes partially dynamic updates.
For each $i=1,\ldots,q$, and $v\in V$ let
$T^\tfrom_{i,v}$ ($T^\tto_{i,v}$) denote
the shortest path tree that the algorithm
of Theorem~\ref{thm:estree} would maintain
for $d=d_{i-1}$ and source $v$ in $G$ (in $\rev{G}$, respectively).
Note that how the trees $T^\tfrom_{i,v}$ and $T^\tto_{i,v}$ evolve 
depends
only on the sequence of updates to $G$ (which, by the oblivious adversary
assumption, does not depend on sets $A_0,\ldots,A_q$ in any way)
and the details of the deterministic algorithm of Theorem~\ref{thm:estree}.
Since only $O(mn)=O(n^3)$ different trees appear in 
$\{T^\tfrom_{i,v}:v\in V\}\cup \{T^\tto_{i,v}:v\in V\}$
throughout all updates,
$A_i$ remains a $(\{T^\tfrom_{i,v}:v\in V\}\cup \{T^\tto_{i,v}:v\in V\},d_{i-1})$-blocker set
throughout the whole sequence of updates with high probability,
by Lemma~\ref{lem:random-subset}.

Let $\treecol^\tfrom_i=\{T^\tfrom_{i,v}:v\in A_{i-1}\}$ and
$\treecol^\tto_i=\{T^\tto_{i,v}:v\in A_{i-1}\}$,
i.e.,
$\treecol^\tfrom_i$ ($\treecol^\tto_i$) contains only trees with roots from a subset $A_{i-1}\subseteq V$.
However $A_i$ being a blocker set of such a collection of trees will turn
out sufficient for our needs.
Clearly, since we have $\treecol^\tfrom_i\cup \treecol^\tto_i\subseteq \{T^\tfrom_{i,v}:v\in V\}\cup \{T^\tto_{i,v}:v\in V\}$,
by the above claim,
$A_i$ in fact remains a $(\treecol^\tfrom_i\cup \treecol^\tto_i,d_{i-1})$-blocker set
throughout the whole sequence of updates with high probability.
 
Now, let $q=\lceil \log_6{n}\rceil$ and for $i=1,\ldots,q$ set $a_i=6^{q-i}$.
To verify whether each $A_i$ remains a $d_i$-hub set subject to partially dynamic updates
to $G$, we proceed as follows.
We deterministically maintain the 
trees $\bigcup_{i=1}^q (\treecol^\tfrom_i\cup \treecol^\tto_i)$
subject to partially dynamic updates to $G$
using Theorem~\ref{thm:estree}.
The total number of changes these trees are subject to throughout the whole sequence of updates is
$$O\left(\sum_{i=1}^q a_{i-1}\cdot m\cdot d_{i-1}\right)=O\left(\sum_{i=1}^q a_{i-1}\cdot m\cdot \frac{n}{a_i}\ln{n}\right)=O\left(nm\log{n}\cdot \sum_{i=1}^q \frac{a_{i-1}}{a_i}\right)=O(nm\log^2{n}).$$

We additionally store each tree 
$T^\tfrom_{i,v}$ (and $T^\tto_{i,v}$), for $v\in A_{i-1}$, in a data structure of
Lemma~\ref{lem:tree-maintain} with $B=A_i$.
Whenever the data structure of Theorem~\ref{thm:estree} updates
some tree, the update is repeated in the corresponding
data structure of Lemma~\ref{lem:tree-maintain}.
Consequently, the total time needed to maintain these additional
data structures is
$O\left(nm\log^2{n}\cdot \sum_{i=1}^q \frac{a_{i-1}}{a_i}\right)=O(nm\log^3{n})$.

After each update we can detect whether each $A_i$
is still a $(\treecol^\tfrom_i\cup \treecol^\tto_i,d_{i-1})$-blocker set
in \linebreak
$O(\sum_i^q|A_{i-1}|\log{n})=O(n\log{n})$ time by querying the relevant
data structures of Lemma~\ref{lem:tree-maintain}\footnote{The Even-Shiloach algorithm (Theorem~\ref{thm:estree}),
apart from maintaining distance labels for all $v\in V$, moves around entire subtrees of the maintained tree $T$.
Hence, in order to ensure that some set $B$ remains a blocker-set of $T$,
  it is not sufficient to simply check whether $B\cap V(T[v])$ whenever the Even-Shiloach algorithm changes
  the distance label of $v$ to $d$ (and, consequently, use a data structure much simpler than that given in Lemma~\ref{lem:tree-maintain}).}
storing $\treecol^\tfrom_i\cup \treecol^\tto_i$.
By Lemma~\ref{lem:hubs-inductive}, a simple inductive argument shows
that if this is the case, each $A_i$ is a $d_i$-hub set of both $G$ and $\rev{G}$.
Hence, verifying all $A_1,\ldots,A_q$ while $G$ evolves
takes
$O(mn\log{n})$
total time.
The algorithm terminates when it turns out that some $A_i$
is no longer a $(\treecol^\tfrom_i\cup \treecol^\tto_i,d_{i-1})$-blocker set.
However, recall that this happens only with low probability,
regardless of whether $A_i$ actually ceases to be a $d_i$ hub set or not.
We have proved the following.

\begin{theorem}\label{t:hubs-exp}
  Let $G$ be an unweighted digraph.
  Let $q=\lceil\log_6{n}\rceil$. For $i=0,\ldots,q$, let $A_i$ be a random $6^{q-i}$-subset of $V$.
  One can maintain the information whether each $A_i$ is a
  $\Theta(6^i\ln{n})$-hub set of $G$, subject edge deletions issued to $G$, in $O(nm\log^3{n})$ total time.
  
  The algorithm might produce false negatives with low probability.
\end{theorem}
By plugging in the hubs of Theorem~\ref{t:hubs-exp} into the algorithms of~\cite{BaswanaHS07, Bernstein16}, we obtain the following.
\begin{corollary}
  Let $G$ be an unweighted digraph. There exists a Las Vegas randomized decremental algorithm
  maintaining exact distance between all pairs of vertices of $G$ with $\Ot(n^3)$ total update time w.h.p.
  It assumes an adversary oblivious to the random bits used.
\end{corollary}

\begin{corollary}
  Let $G$ be an unweighted digraph. There exists a Las Vegas randomized decremental algorithm
  maintaining $(1+\eps)$-approximate distance estimates between all pairs of vertices of $G$
  in $\Ot(nm/\eps)$ total time w.h.p. The algorithm assumes an oblivious adversary.
\end{corollary}

\section{Approximate Shortest Paths for Weighted Graphs}\label{sec:weighted}

In this section we generalize
the reliable hub maintenance algorithms to weighted graphs,
at the cost of $(1+\eps)$-approximation.
First we we give key definitions.

\begin{definition}\label{def:appr-tree}
  Let $G=(V,E)$ be a weighted digraph and let $s\in V$ be a source vertex.
  Let $d$ be a positive integer.
  An out-tree $T\subseteq G$ 
  is called a $(1+\eps)$-approximate shortest path tree from $s$ up to depth~$d$, if
  $T$ is rooted at $s$ and for any $v\in V$ such that $\dist^d_G(s,v)<\infty$,
  we have
  $v\in V(T)$ and $\dist_T(s,v)\leq (1+\eps)\dist^d_G(s,v)$.
\end{definition}
\begin{definition}\label{def:appr-hubs}
  Let $G=(V,E)$ be directed and let $d>0$ be an integer. A set $H_d^\eps\subseteq V$ is
  called an \emph{$(1+\eps)$-approximate $d$-hub set} of $G$ if for every $u,v\in V$ such that
  $\dist_G(u,v)<\infty$, there exists a
  path $P=u\to v$ in $G$ such that $\len(P)\leq (1+\eps)\dist_G(u,v)$ and
  $P$ is $(H_d^\eps,d)$-covered.
\end{definition}

We also extend the definition of a $(T,d)$-blocker set
to trees of depth more than~$d$.
\begin{definition}\label{def:blocker-appr}
  Let $V$ be a vertex set and let $d>0$ be an integer.
  Let $T$ be a rooted tree over~$V$. 
  Define $T^d$ to be the set of all maximal subtrees of $T$ of depth no more than $d$,
  rooted in non-leaf vertices $x\in V(T)$ satisfying $d\mid \treedep{T}{x}$. 
  In other words, the set $T^d$ can be obtained from the tree $T$
  be repeatedly taking a deepest non-leaf vertex $x$ of $T$ such that $\treedep{T}{x}$
  is divisible by $d$, inserting the subtree of $T$ rooted in $x$ into $T^d$
  and consequently removing all the descendants of $x$ out of $T$.
  
  Then, $B$ is a $(T,d)$-blocker set if and only if it is a $(T^d,d)$-blocker set
  (in terms of Definition~\ref{def:blocker}).
  Let $\treecol$ be a collection of rooted trees over $V$.
  We call $B$ a $(\treecol,d)$-blocker set if and only if $B$ is a $(T,d)$-blocker for
  each $T\in\treecol$.
\end{definition}
We now state the main theorem relating blocker sets in $(1+\eps)$-approximate
shortest path trees to the approximate hub sets. The theorem is proved in Section~\ref{sec:weighted-hubs-simple-proof}.
\begin{restatable}{theorem}{weightedhubssimple}\label{thm:weighted-hubs-simple}
Let $G=(V,E)$ be a directed graph and let $d<n$ be an even integer.
   Let $\treecol^{\tfrom}=\{T^{\tfrom}_v: v\in V\}$
  ($\treecol^{\tto}=\{T^{\tto}_v: v\in V\}$)
  be a collection of $(1+\eps)$-approximate shortest
  path trees up to depth-$3d$ from all vertices in $G$ (in $\rev{G}$, resp.).
 
  Let $B\subseteq V$ be a $(\treecol^{\tfrom}\cup \treecol^{\tto},\frac{d}{2})$-blocker set.
  Then $B$ is a $(1+\eps)^p$-approximate $2dp$-hub set of $G$,
  where $p=\lceil\log_2{n}\rceil+1$.
\end{restatable}

Finally, we explain how to incorporate these tools into our improved dynamic APSP
algorithms in order to generalize then to weighted graphs.
Recall that our reliable hubs maintenance algorithms for unweighted graphs essentially maintained
some shortest path trees up to depth $d$ and either computed their blocker sets
using King's algorithm, or dynamically verified whether the sampled hub sets
remain blocker sets of the shortest path trees.

We first replace all shortest path trees up to depth $d$ with $(1+\eps)$-approximate
shortest path trees up to depth $d$.
We use the following extension of Bernstein's $h$-SSSP algorithm.
\begin{restatable}{lemma}{hsspext}\label{lem:hsssp-ext}
  The $h$-SSSP algorithm of Theorem~\ref{thm:hsssp} can be extended so that it maintains
  a $(1+\eps)$-approximate shortest path tree up to depth $h$ from $s$
	within the same time bound.
\end{restatable}
\begin{proof}
  The $h$-SSSP data structure (see \cite{Bernstein16}) actually maintains $O(\log(nW))$
  $h$-SSSP$_k$ data structures, handling different ranges of the distances
  to be estimated.
  For each $k=0,\ldots,\lfloor\log_2{nW}\rfloor$, the $h$-SSSP$_k$ data structure
  maintains an out-tree $T_k\subseteq G$ rooted at $s$, such that for all $v\in V$,
  if $2^k\leq \dist_G^d(s,v)\leq 2^{k+1}$,
  then $v\in V(T_k)$ and $\dist_G(s,v)\leq \dist_{T_k}(s,v)\leq (1+\eps)\dist_G^d(s,v)$.
  The final $(1+\eps)$-approximate estimate $\dist'(s,v)$ of $\dist_G^d(s,v)$ produced by $h$-SSSP is defined
  as $\dist'(s,v)=\min_k\{\dist_{T_k}(s,v):v\in V(T_k)\}$.
  Even though each $h$-SSSP$_k$ has total update time $O(mh/\eps+\Delta)$,
  Bernstein \cite{Bernstein16} uses additional data structures and tricks to make different $h$-SSSP$_k$
  components register only relevant edge updates and make the
  dependence on $\Delta$ only $O(\Delta)$.

  Nevertheless, all $h$-SSSP$_k$ components maintain their trees $T_k$ explicitly
  in the stated \linebreak $O(mh/\eps\log{n}\log{nW}+\Delta)$ total time.
  Having a single out-tree instead of $\lfloor\log{nW}\rfloor+1$ will simplify
  our further developments considerably.
  Therefore, we combine these trees $T_k$ into a single $(1+\eps)$-approximate  
  shortest path tree $T$ up to depth $d$ from $s$, as follows.

  The $h$-SSSP algorithm \cite{Bernstein16} has an additional property that it explicitly maintains,
  for each vertex $v\neq s$, both the value $\dist'(s,v)$ and the index $k_v$ of the tree $T_{k_v}$ such
  that
  $\dist'(s,v)=\dist_{T_{k_v}}(s,v)=\min_k\{\dist_{T_k}(s,v):v\in V(T_k)\}$.
  Whenever $k_v$ is updated for some $v$,
  we make $\tpar_{T_{k_v}}(v)$ the new parent
  of $v$ in $T$ and the edge $\tpar_{T_{k_v}}(v)v=e\in E(T_{k_v})$ the only
  incoming edge of $v$ in $T$.
  When $\dist'(s,v)$ becomes $\infty$, $v$ is removed from $T$.

  We first prove that by proceeding this way, $T$ actually remains a tree, or, in other words,
  the edges of the form $\tpar_T(v)v$, where
  $v\in V(T)\setminus\{s\}$, form an out-tree that is a subgraph of $G$.
  Since each vertex in $T$ except of $s$ (which has $0$ outgoing edges) has exactly one incoming edge,
  we only need to prove that $T$ has no directed cycles.
  For contradiction, suppose there is a cycle $e_1\ldots e_k$ in $T$,
  where $e_i=u_iu_{i+1}$, $u_{k+1}=u_1$ and $u_i=\tpar_T(u_{i+1})$ for all $i=1,\ldots k$.
  Consider some edge $e_i$. 
  We have $u_i=\tpar_T(u_{i+1})=\tpar_{T_{k_{u_{i+1}}}}(u_{i+1})$ and hence
  $\dist'(s,u_{i+1})=\dist_{T_{k_{u_{i+1}}}}(s,u_{i+1})=\dist_{T_{k_{u_{i+1}}}}(s,u_i)+\wei_{T_{k_{u_{i+1}}}}(u_iu_{i+1})$ for some $l$.
  Recall that $G$ has positive edge weights and $T_{k_{u_{i+1}}}\subseteq G$.
  As a result
  $\dist'(s,u_{i+1})>\dist_{T_{k_{u_{i+1}}}}(s,u_i)\geq \dist'(s,u_i)$.
  This way we obtain $\dist'(s,u_1)=\dist'(s,u_{k+1})>\dist'(s,u_k)>\ldots>\dist'(s,u_1)$, a contradiction.
  
  Second, we prove that for all $v\in V(T)$ we have $\dist_T(s,v)\leq \dist'(s,v)$.
  We proceed by induction on $\treedep{T}{v}$.
  If $s=v$, this is clearly true.
  Otherwise, let $\treedep{T}{v}>0$ and suppose the claim
  holds for all vertices of $T$ of smaller depth.
  Let $p=\tpar_T(v)$.
  By the inductive hypothesis, we have $\dist_T(s,p)\leq \dist'(s,p)$.
  By the definition of $T$, we also have $p=\tpar_{T_{k_v}}(v)$ and $\wei_T(p,v)=\wei_{T_{k_v}}(p,v)$.
  Hence,
  $$\dist'(s,v)=\dist_{T_{k_v}}(s,p)+\wei_{T_{k_v}}(p,v)\geq \dist'(s,p)+\wei_T(p,v)\geq \dist_T(s,p)+\wei_T(p,v)=\dist_T(s,v).$$

  Finally, note that as the $h$-SSSP algorithm guarantees that $\dist'(s,v)\leq(1+\eps)\dist_G^d(s,v)$,
  we also have $\dist_T(s,v)\leq (1+\eps)\dist_G^d(s,v)$.
  Clearly, as $T\subseteq G$, for all $v\in V(T)$ we have $\dist_T(s,v)\geq \dist_G(s,v)$.
\end{proof}

By Theorem~\ref{thm:weighted-hubs-simple}, by finding blocker sets of approximate shortest path trees
(as in Definition~\ref{def:blocker-appr}), we can 
compute/verify 
$(1+\eps')^{\Theta(\log{n})}$-approximate $\Theta(d\log{n})$-hub sets
as before.

Given appropriate hub sets,
all that both our deterministic incremental $(1+\eps)$-approximate APSP algorithm,
and Bernstein's
randomized $(1+\eps)$-approximate partially dynamic APSP algorithm do,
is essentially set up and maintain a ``circuit'' (i.e., a collection of data structures
whose outputs constitute the inputs of other structures)
of $h$-SSSP data structures from
the hubs with different parameters $h$ and appropriately set $\eps'$.
In order to make these algorithms work with our reliable approximate
hub sets, we basically need to play with the parameters: increase all $h$'s
by a polylogarithmic factor, and decrease $\eps'$ by a polylogarithmic factor.
We discuss the details in Appendix~\ref{s:reliable-weighted}.

\subsection{Proof of Theorem~\ref{thm:weighted-hubs-simple}}\label{sec:weighted-hubs-simple-proof}

\begin{lemma}\label{lem:tree-path-blocker}
  Let $T$ be a rooted tree over $V$.
  Let $d$ be a positive integer and suppose $B\subseteq V$ is a $(T,d)$-blocker set (with respect to Definition~\ref{def:blocker-appr}).
  Then for any $v\in V(T)$, the path $\treepath{T}{v}$ is $(B,2d)$-covered.
\end{lemma}
\begin{proof}
  By the definition of $B$,
  $\treepath{T}{v}$ can be expressed as $P_1\ldots P_k$, $P_i=u_i\to v_i$,
  where $\hops{P_i}=d$ and $V(P_i)\cap B\neq\emptyset$ for $i=1,\ldots,k-1$ and $\hops{P_k}\leq d$.
  This is because each for such $i$ we have $d\mid \treedep{T}{u_i}$
  and $P_i$ is a root-leaf path in some tree of the set $T^d$ obtained from $T$
  using Definition~\ref{def:blocker-appr}.
  Since each $P_i$ is trivially $(B,d)$-covered and $P_k$ is the only subpath
  that may not contain a vertex of $B$, by Lemma~\ref{lem:cover-concat}
  we conclude that $\treepath{T}{v}$ is indeed $(B,2d)$-covered.
\end{proof}

\begin{lemma}\label{lem:aux-appr-hub}

  Let $G=(V,E)$ be a directed graph and let $d>0$ be an even integer.
  Let $H\subseteq V$.
  Let $\treecol^{\tfrom}=\{T^{\tfrom}_v: v\in H\}$
  ($\treecol^{\tto}=\{T^{\tto}_v: v\in H\}$)
  be a collection of $(1+\eps)$-approximate shortest
  path trees up to depth-$3d$ from $H$ in $G$ (in $\rev{G}$, resp.).
  

  Let $B\subseteq V$ be a $(\treecol^{\tfrom}\cup \treecol^{\tto},\frac{d}{2})$-blocker set.
  Then, for any path $P=u\to v$ in $G$ that is $(H,d)$-covered 
  and integer $j\geq 0$, if $\hops{P}\leq 2^jd$,
  then there exists a path $P'=u\to v$ in $G$ such that
    $\len(P')\leq (1+\eps)^j\cdot \len(P)$ and
    $P'$ is $(B,2d(j+1))$-covered.
\end{lemma}
\begin{proof}
    We proceed by induction on $j$. For $j=0$, $j+1=2^j$ so $P'=P$ has the
    required properties since $P$ has no more than $d\leq 2d$ edges.
    
    Let $j>0$ and suppose the lemma holds for all $j'<j$. 
    If $\hops{P}\leq d\cdot 2^{j-1}$, then $P'$
    exists by the inductive hypothesis.
    Assume $\hops{P}> d\cdot 2^{j-1}= d$.

    By Lemma~\ref{lem:hubs-atleast}, we can 
    partition $P$ into $s\geq 1$ subpaths $P_1,\ldots,P_s$,
    where $P_i=u_i\to v_i$, such that
    $\hops{P_i}\in [d,3d]$ for all $i=1,\ldots,s$ and $u_i\in H$ for all $i=2,\ldots,s$.

	Since $P_i$ has no more than $3d$ edges, for each $i\in \{2,\ldots,s\}$ we have
    $$\len(T^{\tfrom}_{u_i}[v_i])=\dist_{T^{\tfrom}_{u_i}}(u_i,v_i)\leq (1+\eps)\disth{3d}_G(u_i,v_i)\leq (1+\eps)\len(P_i),$$
    whereas for $i=1$ we have
    $$\len(T^{\tto}_{v_1}[u_1])=\dist_{T^{\tto}_{v_1}}(v_1,u_1)\leq (1+\eps)\disth{3d}_G(u_1,v_1)\leq (1+\eps)\len(P_1).$$

    For $i=1$, let $Q_1=\rev{(\treepath{T^{\tto}_{v_1}}{u_1})}$, whereas for $i=2,\ldots,s$, set
    $Q_i=\treepath{T^{\tfrom}_{u_i}}{v_i}$.
    Then we clearly have $\len(Q_i)\leq (1+\eps)\len(P_i)$ for all $i$.

    Let $Q_1',\ldots,Q_r'$ be a partition of $Q_1\ldots Q_{s}$ into blocks, i.e.,
    $Q_i'=Q_{t_i}Q_{t_i+1}\ldots Q_{t_{i+1}-1}$,
    where $1=t_1<t_2<\ldots<t_{r}<t_{r+1}=s+1$,
    such that:
    \begin{itemize}
      \item If $\hops{Q_i}\geq\frac{d}{2}$, then $Q_i$ forms a single-element block $Q_z'$.
        In other words, for some $z$, $t_z=i$ and $t_{z+1}=i+1$.
      \item If $Q_x,\ldots,Q_y$ is a \emph{maximal} set of paths such that for all $g\in [x,y]$,
        we have $\hops{Q_g}<\frac{d}{2}$, then $Q_x,\ldots,Q_y$ form a single block $Q_z'$.
        In other words, for some $z$, $t_z=x$ and $t_{z+1}=y+1$.
    \end{itemize}
    We now set $P'=Q_1''Q_2''\ldots Q_r''$, where each $Q_z''$, for $z=1,\ldots,r$,
    is defined as follows.
    \begin{itemize}
      \item If $Q_z'=Q_i$, where $\hops{Q_i}\geq\frac{d}{2}$, then $Q_z''=Q_z'=Q_i$.
        Note that by Lemma~\ref{lem:tree-path-blocker}, $Q_z''$ is $(B,d_z)$-covered, where $d_z=2\cdot\frac{d}{2}=d$.
        Moreover, since $\hops{Q_z}\geq\frac{d}{2}$, $Q_z''$ necessarily contains a vertex from~$B$.

  \item If $Q_z'=Q_{t_z}\ldots Q_{t_{z+1}-1}$, where for each $g\in [t_z,t_{z+1}-1]$, $\hops{Q_g}<\frac{d}{2}$,
        then $\hops{Q_z'}\leq \frac{\hops{P}}{2}\leq d\cdot 2^{j-1}$.
        Since $Q_z'$ is a concatenation of paths of length less than $\frac{d}{2}$,
        each of which has a vertex of $H$ as its endpoint,
        by Lemma~\ref{lem:cover-concat}
        $Q_z'$ is $(H,d)$-covered.
        Hence, we can apply the inductive hypothesis to $Q_z'$:
        there exists a path $Q_z''=u_{t_z}\to v_{t_{z+1}-1}$ such that $\len(Q_z'')\leq (1+\eps)^{j-1}\cdot\len (Q_z')$
        and $Q_z''$ is $(B,d_z)$-covered, where $d_z=2dj$.
    \end{itemize}
    By Lemma~\ref{lem:cover-concat}, $P'$ is $(B,D)$-covered, where
    $$D=\max\left\{d_x+\ldots+d_z:1\leq x\leq z\leq r\text{ and } V(Q_y'')\cap B=\emptyset\text{ for all }y\in (x,z)\right\}.$$
    
    We now prove that $D\leq 2d(j+1)$.
    Let $x,z$ be such that $1\leq x\leq z\leq r$ and $V(Q_y'')\cap B=\emptyset$ for all $y\in (x,z)$.
    If $(x,z)\cap\mathbb{Z}=\emptyset$, then $z\leq x+1$.
    Otherwise, let $g\in (x,z)\cap\mathbb{Z}$. Observe that $V(Q_g'')\cap B=\emptyset$ can only hold if $Q_g''$ is a \emph{maximal}
    set of consecutive paths $Q_i$, $i\leq r$, with less than $\frac{d}{2}$ edges
    each.
    Hence, by the aforementioned maximality, $V(Q_{g-1}'')\cap B\neq\emptyset$
    and $V(Q_{g+1}'')\cap B\neq\emptyset$.
    We thus have $x\geq g-1$ and $y\leq g+1$.
    Consequently, we obtain $z\leq x+2$ and from $g\in (x,y)$
    we obtain $z=x+2$ and $g=x+1$.
    We clearly have $d_x+d_{x+1}+d_{x+2}=d_x+d_g+d_z=d+2dj+d=2d(j+1)$.

    Observe now that among at most two consecutive paths $Q_i''$,
    there can be at most one \emph{maximal}
    set of consecutive paths $Q_i$, $i\leq r$, with less than $\frac{d}{2}$ edges
    each.
    Therefore, $d_x+d_{x+1}\leq 2dj+d\leq 2d(j+1)$ for any $x\leq r$.

    It remains to prove that $\len(P')\leq (1+\eps)^j\len(P)$. We have:
    \begin{align*}
      \len(P')&=\sum_{i=1}^r\len(Q_i'')\\
              &\leq (1+\eps)^{j-1}\sum_{i=1}^r\len(Q_i')\\
              &=(1+\eps)^{j-1}\sum_{i=1}^{s}\len(Q_i)\\
              &\leq (1+\eps)^{j-1}\sum_{i=1}^{s}(1+\eps)\len(P_i)\\
              &=(1+\eps)^j\len(P).\qedhere
    \end{align*}

\end{proof}

\begin{corollary}\label{cor:hubs-inductive-appr}
Let $G=(V,E)$ be a directed graph and let $d>0$ be an even integer.
Let $H_d^\eps$ be a $(1+\eps)^q$-approximate
  $d$-hub set of $G$.

   Let $\treecol^{\tfrom}=\{T^{\tfrom}_v: v\in H_d^\eps\}$
  ($\treecol^{\tto}=\{T^{\tto}_v: v\in H_d^\eps\}$)
  be a collection of $(1+\eps)$-approximate shortest
  path trees up to depth-$3d$ from vertices of $H_d^\eps$ in $G$ (in $\rev{G}$, resp.).
 
  Let $B\subseteq V$ be a $(\treecol^{\tfrom}\cup \treecol^{\tto},\frac{d}{2})$-blocker set.
  Then $B$ is a $(1+\eps)^{p+q-1}$-approximate $2dp$-hub set of $G$,
  where $p=\lceil\log_2{n}\rceil+1$.
\end{corollary}

\begin{proof}
  To prove that $B$ is in fact a $(1+\eps)^{p+q-1}$-approximate
  $2dp$-hub set for $G$, let $u,v\in V$ be such that $\dist_G(u,v)<\infty$
  and take any path $P$ from $u$ to $v$ that is
  $(H_d^\eps,d)$-covered and $\len(P)\leq (1+\eps)^q\dist_G(u,v)$.
  As $\hops{P}\leq n$ and $d\geq 1$, $P$ clearly has less than $d2^{p-1}$ edges.
  
  Therefore, by Lemma~\ref{lem:aux-appr-hub} (for $H=H^\eps_d$), there exists a path $P'=u\to v$
  in $G$ such that $\len(P')\leq (1+\eps)^{p-1}\len(P)\leq (1+\eps)^{p+q-1}\dist_G(u,v)$
  and $P'$ is $(B,2dp)$-covered.
\end{proof}

In order to prove Theorem~\ref{thm:weighted-hubs-simple}, 
just apply Corollary~\ref{cor:hubs-inductive-appr} with $q=1$.

\bibliographystyle{plainurl}
\bibliography{references2}

\newpage

\appendix

\section{Deterministic Incremental Algorithm for Dense Graphs}\label{sec:simple-incremental}
Let us focus on the restricted version of the problem -- this assumption can be dropped
by applying Lemma~\ref{lem:op-reduction}. So, $\Delta=O(n^2\log{W}/\eps)$.

Let use $k=\lceil \log_2 n \rceil$ layers of $n$ $h$-SSSP data structures (see Theorem~\ref{thm:hsssp}) with $h=2$.
Let $\eps_1=\eps/2k$.
For $i=1,\ldots,k$, and $v\in V$ the data structure $D_{i,v}$ will maintain
distance estimates $\dist'_{i,v}(w)$, satisfying
\begin{equation}\label{e:tmp}
\dist_G(v,w)\leq \dist'_{i,v}(w)\leq (1+\eps_1)^{i-1}\dist^{2^i}_G(v,w).
\end{equation}
Clearly, since $\dist^{2^k}_G(v,w)=\dist_G(v,w)$ for all $v,w\in V$,
data structures $D_{k,v}$ will maintain $(1+\eps_1)^{k}\leq (1+\eps)$ approximate
distances from $v$ to all other vertices.

We achieve that by making setting up $D_{i,v}$ to be a $h$-SSSP data structure
with $h=2$ run from $v$ on:
\begin{itemize}
  \item either the input graph $G$ if $i=1$,
  \item or a complete graph $G_i$ on $V$ with edge weights equal to the estimates produced by the previous
    layer, i.e., $\wei_{G_i}(uv)=\dist'_{i-1,u}(v)$. 
\end{itemize}
Then, inequality~(\ref{e:tmp}) follows easily by induction.

Recall that for dense graphs, the $h$-SSSP algorithm~\cite{Bernstein16} runs in $O(n^2h\log{n}\log{(nW)}/\eps+\Delta)$ time
and updates some of its estimates $O(nh\log{(nW)}/\eps)$ times.
In our case, whenever $D_{i-1,u}$ updates $\dist'_{i-1,u}(v)=\wei_{G_i}(uv)$, we
send that update to all data structures $D_{i,\cdot}$ run on $G_i$.
Let $\Delta_{i,v}$ be the number of times $D_{i,v}$ updates some of its estimates.
Then, $\Delta_{i,v}=O(n\log{(nW)}/\eps_1)$ and thus the total running time of the algorithm
can be seen to be
$$O\left(\sum_{v\in V}\Delta+\sum_{i=1}^k \sum_{v\in V}n^2\log{n}\log{(nW)}/\eps_1+\sum_{i=2}^k\sum_{v\in V}\sum_{w\in V}\Delta_{i-1,w}\right)=O(n^3\log^3{n}\log{(nW)}/\eps).$$

\section{Details of the State-Of-the-Art Decremental APSP Algorithms}\label{s:baswanabernstein}

\subsection{The Algorithm of Baswana et al. \cite{BaswanaHS07} Given Hub Sets}
The algorithm of Baswana et al. \cite{BaswanaHS07} explicitly maintains
the matrix of pairwise distances between the vertices of an \emph{unweighted} digraph
under edge deletions in $O(n^3\log^2{n})$ total time.
The algorithm is Monte Carlo randomized.
Below we briefly describe how to turn it 
into a Las Vegas algorithm running in time $O(n^3\log^2{n}+nm\log^3{n})=O(n^3\log^3{n})$.

Let the sets $A_0,\ldots,A_q$ be defined as before (i.e., $a_i=6^{q-i}$).
The algorithm of Baswana et al.~\cite{BaswanaHS07} can be rephrased
using our terminology as follows.
The set $A_i$ is used to maintain the values $\dist_i'(u,v)$ such that
$\dist_i'(u,v)\geq \dist(u,v)$ and additionally $\dist_i'(u,v)=\dist(u,v)$
if $\dist(u,v)\in [t_i+1,t_{i+1}]$.
The thresholds $t_i$ are defined as follows: $t_0=-1$, $t_i=\min(d_i,n-1)$ for $i=1,\ldots,q$
and $t_{q+1}=n-1$.

Clearly, the intervals $[t_i+1,t_{i+1}]$ cover the entire range $[0,n-1]$
of possible distances, so for any $u,v\in V$ we have $\min_{i=0}^q \{\dist'_i(u,v)\}=\dist(u,v)$.

It remains to show how to maintain the values $\dist'_i(u,v)$.
For each $i=0,\ldots,q$ and $v\in A_i$ we maintain
shortest path trees $T^\tfrom_{i,v}$ and $T^\tto_{i,v}$  up to depth $t_{i+1}$
in $G$ and $\rev{G}$,
respectively.
We keep each $\dist'_i(u,v)$ equal to
$$\dist'_i(u,v)=\min_{a\in A_i}\left\{\dist_{T^\tto_{i,a}}(a, u)+\dist_{T^\tfrom_{i,v}}(a,v)\right\}.$$

To see that this is correct, note that since $A_0=V$ and $A_i$ is a $d_i$-hub set, then
for each $u,v$ such that $\dist(u,v)\in [t_i+1,t_{i+1}]$, $\dist(u,v)>d_i$
and hence there exists
some shortest $P=u\to v$ path that goes through some vertex of $a\in A_i$.
Equivalently, $P=P_1P_2$, where $P_1=u\to a$ and $P_2=a\to v$ and both
$P_1$ and $P_2$ are shortest and $\hops{P_i}\leq\hops{P}\leq t_{i+1}$.
Consequently, $\hops{P}=\hops{P_1}\cup\hops{P_2}=\dist_{T^\tto_{i,a}}(a, u)+\dist_{T^\tfrom_{i,v}}(a,v)$.

To keep $\dist'_i(u,v)$ updated, whenever some 
$\dist_{T^\tto_{i,a}}(a, u)$ (or $\dist_{T^\tfrom_{i,v}}(a,v)$) changes we go
through all $v\in V$ ($u\in V$, resp.), and set
$\dist'_i(u,v):=\min(\dist'_i(u,v),\dist_{T^\tto_{i,a}}(a, u)+\dist_{T^\tfrom_{i,v}}(a,v))$.
Observe that each 
$\dist_{T^\tto_{i,a}}(a, u)$ (similarly ($\dist_{T^\tfrom_{i,v}}(a,v)$)
changes at most $t_{i+1}=O(d_i)$ times
and each such changes is followed by a computation that runs in $O(n)$ time.
The total time needed for updating
the values $\dist'_i(u,v)$ (apart from maintaining the trees themselves) is hence
$$O\left(\sum_{i=1}^q |A_i|\cdot n\cdot d_i\cdot n\right)=O\left(\sum_{i=1}^q n^3\log{n}\right)=O(n^3\log^2{n}).$$
The total time needed to maintain the needed shortest path trees
is also 
$$O\left(\sum_{i=1}^q |A_i|\cdot md_i\right)=O\left(\sum_{i=1}^q mn\log{n}\right)=O(mn\log^2{n})=O(n^3\log^2{n}).$$

Combining with our hub set verification procedure, we obtain a Las Vegas
algorithm running in $O(n^3\log^2{n}+mn\log^3{n})$ time with high probability.

\subsection{The Algorithm of Bernstein \cite{Bernstein16} Given Hub Sets}\label{sec:bernstein}

Bernstein \cite{Bernstein16} showed that $(1+\eps)$-approximate distance
estimates between all pairs of vertices of an unweighted\footnote{Bernstein \cite{Bernstein16} showed that
even for weighted graphs, but in this section we concentrate on the unweighted case.
We explain the modifications needed for the weighted version in Section~\ref{sec:weighted}.}
digraph $G$ subject
to partially dynamic updates can be maintained
in $O(mn\log^5{n}/\eps)$ total time.
The algorithm is Monte Carlo randomized
and assumes an oblivious adversary.

Given the sets $A_0,\ldots,A_q$, the algorithm works as follows.
For each subsequent $i=q,\ldots,0$, for $(u,v)\in (A_i\times V)\cup (V\times A_i)$
it maintains
estimates $\dist'_i(u,v)$ 
satisfying $$\dist_G(u,v)\leq \dist'_i(u,v)\leq (1+\eps')^{q-i+1} \dist_G(u,v).$$

For $i=q$, the values $\dist'_q(u,v)$ are maintained using
$h$-SSSP data structures of Theorem~\ref{thm:hsssp}
run on $G$ and $\rev{G}$ from the only vertex of $A_q$
for $h=n$.
This takes $O(mn\log^2{n}/\eps')$ time.

For $i<q$ on the other hand, the distance estimates $\dist_i'(u,\cdot)$
and $\dist'_i(\cdot,u)$, where $u\in A_i$, are maintained using
$h$-SSSP data structures of Theorem~\ref{thm:hsssp} run
for $h=d_{i+1}+1$
on $G_{i,u}$ and $\rev{G_{i,u}}$, where $G_{i,u}=G\cup S_{i,u}$,
$S_{i,u}=(V,\{uv:v\in A_{i+1}\})$, and $\wei_{S_{i,u}}(uv)=\dist'_{i+1}(u,v)$.
Using the assumption that $\dist'_{i+1}(x,y)\leq (1+\eps')^{q-i}\dist_G(x,y)$ which holds for $x\in V$ and $y\in A_{i+1}$,  one can show that $\dist_{G_{i,u}}^{d_{i+1}+1}(u,v)\leq (1+\eps')^{q-i}\dist_G(u,v)$
for $v\in V$ and hence $\dist'_i(u,v)\leq (1+\eps')^{q-i+1}\dist_G(v,u)$.
Consequently, the total cost of maintaining estimates
$\dist_i'(u,\cdot)$ and $\dist_i'(\cdot,u)$ is
$O(a_i\cdot d_{i+1}m\log^2{n}/\eps')$.
As we set $a_i=6^{q-i}$ and we have $d_{i+1}=O(n/a_i)$, this total cost
is in fact $O(nm\log^3{n}/\eps')$.

In total, as $q=O(\log{n})$ and we set $\eps'=\Theta(\eps/\log{n})$,
the total update time is $O(nm\log^5{n}/\eps)$.
Hence, combining Bernstein's algorithm with our verification procedure
yields a Las Vegas algorithm with asymptotically the same running time with
high probability.

\section{Using Reliable Hubs in Algorithms for Weighted Digraphs}\label{s:reliable-weighted}

\subsection{Deterministic Incremental APSP}
In order to adjust the data structure of Section~\ref{sec:sparse-incremental}
so that it supports edge weights, we do the following.
Assume wlog. that we solve the restricted problem (see Section~\ref{sec:preliminaries}).
First, instead of maintaining exact shortest paths
between nearby vertices in $G$, we
maintain collections $\treecol^{\tfrom}$ and $\treecol^{\tto}$
of $(1+\eps')$-approximate shortest path
trees up to depth-$3d$ from all vertices of $V$ in $G$ and $\rev{G}$ respectively.
To this end, we use data structures of Theorem~\ref{thm:hsssp}.
Consequently, the total time spent in this component
is
$O(mnd\log{n}\log{(nW)}/\eps')$,
as only $O(m\log{W}/\eps')$
updates are issued to any $h$-SSSP data structure maintaining the approximate
shortest path trees.

Next, the $d$-hub set $H_d$ is replaced with a $(1+\eps')^{p}$-approximate
$d'$-hub set $H_{d'}$, also of size $O(\frac{n}{d}\log{n})$, where $d'=\Theta(d\log{n})$
and $p=\Theta(\log{n})$.
By Theorem~\ref{thm:weighted-hubs-simple}, such a hub set can be computed deterministically at the beginning
of each phase in $O(n^2\log{n})$ time using Lemma~\ref{lem:king}.
An analogue of Lemma~\ref{lem:incremental-hubs} also
holds for approximate hub sets and therefore whenever $uv$ is added to the
graph of the weight of $uv$
is decreased, $\{u,v\}$ is added to $H_{d'}^\eps$.

The approximate shortest paths between the hubs are maintained
analogously as in Section~\ref{sec:sparse-incremental}, also
using Theorem~\ref{thm:dense-incremental},
in
$O\left(\frac{md}{n\log{n}}|H_{d'}|^3\log^3{n}\log{(nW)}/\eps'\right)=O\left(m\left(\frac{n}{d}\right)^2\log^5{n}\log{(nW)}/\eps'\right)$
total time.
However, since $H_{d'}$ is now only a $(1+\eps')^{p}$-approximate
$d'$-hub set, we need to explicitly maintain the estimates of
distances $\dist_G^{d'}(u,v)$ for $u,v\in H_{d'}$
(the algorithm already maintains approximate shortest path trees up to depth $3d$, but $d' = \Theta(d \log n)$.)
For this, we need $O(n)$ data structures of Theorem~\ref{thm:hsssp},
run with $h=d'$.
Throughout the whole update sequence, these data structures
require
$O(mnd'\log{n}\log{(nW)}/\eps')=O(mnd\log^2{n}\log{(nW)}/\eps')$ time.
Observe that this component produces $(1+\eps')$-approximate
distance estimates between pairs of hubs $H_{d'}$.

The remaining two components producing the estimates between
$H_{d'}\times V$ and subsequently $V\times V$, 
are also analogous to those
in Section~\ref{sec:sparse-incremental}. The only
change is that $H_{d'}$ is now a $(1+\eps)^p$-approximate $d'$-hub
set, instead of an (exact) $d$-hub set, so the parameter $h$
in the used data structures of Theorem~\ref{thm:hsssp} needs
to be appropriately increased.
The total time used by these components is
again $O(mnd\log^2{n}\log{(nW)}/\eps')$.
The final estimates are $(1+\eps')^{p+4}$-approximate,
so we set $\eps'=\eps/(2p+8)$ in order to obtain
$(1+\eps)$-approximate distance estimates.

\begin{theorem}
  Let $G$ be a weighted directed graph. There exists a deterministic incremental
  algorithm maintaining $(1+\eps)$-approximate distance estimates between all
  pairs of vertices of $G$ in \linebreak
  $O(mn^{4/3}\log^{4}{n}\log{(nW)}/\eps+\Delta)$
  total time.
\end{theorem}
\begin{proof}
  The total update time is
  $$O\left(m\left(\frac{n}{d}\right)^2\log^5{n}\log(nW)/\eps'+mnd\log^2{n}\log(nW)/\eps'\right)=O\left(mn\left(\frac{n}{d^2}\log^3{n}+d\right)\log^3{n}\log(nW)/\eps\right).$$
  Hence, the described algorithm solving the restricted version
  of the problem is asymptotically the most
  efficient if $d=n^{1/3}\log{n}$.
  By applying Lemma~\ref{lem:op-reduction}, we obtain the desired bound.
\end{proof}

\subsection{Bernstein's Partially-Dynamic APSP}

Let $\alpha=2(\lceil\log_2{n}\rceil+1)$ and set $q=\lceil\log_{\alpha}{n}\rceil$.
Let $a_0=n$ and $a_i=\alpha^{q-i}$ for $i=1,\ldots,q$.
By increasing $n$ by at most a factor of $2$, we may assume
that $a_{i+1}\mid a_i$ for all $i>1$.
Again $p=\lceil\log_2{n}\rceil+1$, so $\alpha=2p$.
Similarly as in Section~\ref{sec:verification}, let $A_i$ be a random
$a_i$-subset of $V$
and let $d_i=z\cdot \frac{n}{a_{i+1}}\lceil\ln{n}\rceil$,
where $z$ is such a constant that, with high probability,
$A_i$ is a $(\treecol,z\cdot \frac{n}{2a_{i}}\lceil\ln{n}\rceil)=(\treecol,\frac{d_{i-1}}{2})$-blocker
set for any fixed collection $\treecol$ of trees over $V$ such that $|\treecol|=O(mn/\eps')$,
where $\eps'$ is to be set later.

We proceed similarly as in Section~\ref{sec:verification},
and maintain, for each $i=1,\ldots,q$, collections $\treecol^{\tfrom}_i$ and $\treecol^{\tto}_i$
of depth-$3d_{i-1}$ $(1+\eps')$-approximate
shortest path trees from all vertices of $A_{i-1}$ in both $G$ and $\rev{G}$.
We now prove
that as long as $A_i$ remains a $(\treecol^\tfrom_i\cup\treecol^\tto_i,\frac{d_{i-1}}{2})$-blocker set, it is
a $(1+\eps')^{ip}$-approximate $d_i$-hub set of both $G$ and $\rev{G}$.

Clearly, since $A_0=V$, $A_0$ is a $(1+\eps')^0$ approximate $d_0$ hub set of both $G$ and $\rev{G}$.
Now suppose $i\geq 1$.
We apply Corollary~\ref{cor:hubs-inductive-appr} with $H^\eps_d=A_{i-1}$, $q=(i-1)p$, 
$d=d_{i-1}$, $\treecol^\tfrom=\treecol^\tfrom_i$, $\treecol^\tto=\treecol^\tto_i$, and $B=A_i$,
and obtain that $A_i$ is a $(1+\eps')^{p+(i-1)p}=(1+\eps')^{ip}$-approximate $2d_{i-1}p=d_ip\cdot \frac{2a_{i+1}}{a_i}=d_i\frac{2p}{\alpha}=d_i$-hub set
of both $G$ and $\rev{G}$.

If we additionally store the individual trees of
$\treecol^{\tfrom}_i\cup \treecol^{\tto}_i$
in data structures of Lemma~\ref{lem:tree-maintain} (first
splitting the trees into maximal subtrees of depth $\frac{d_{i-1}}{2}$,
as in Definition~\ref{def:blocker-appr}),
then the hub sets $A_i$ can be verified
in
\begin{align*}
  O\left(\sum_{i=1}^q a_{i-1}\cdot md_{i-1}\log^2{n}\log{(nW)}/\eps'\right)&=O\left(nm\log^3{n}\log{(nW)}\sum_{i=1}^q\frac{a_{i-1}}{a_i}/\eps'\right)\\
  &=O(nm\log^5{n}\log{(nW)}/\eps')
\end{align*}
total time.

Let $x_i$ be such that Bernstein's algorithm maintains the $(1+\eps')^{x_i}$-approximate
distance estimates $\dist'_i(u,v)$ for $(u,v)\in (A_i\times V)\cup (V\times A_i)$
using the $(1+\eps')^{x_{i+1}}$-approximate distance estimates $\dist'_{i+1}(\cdot,\cdot)$
and the $(1+\eps')^{(i+1)p}$-approximate hub set $A_{i+1}$.
One can easily show
that the estimates $\dist'_i(\cdot,\cdot)$ are in fact
$(1+\eps')^{x_{i+1}+(i+1)p+1}$-approximate.
Hence, we have $x_q=1$ and $x_i=x_{i+1}+(i+1)p+1$ for all $i<q$,
so we conclude that $x_i=1+\sum_{j=i+1}^q jp$
and therefore $x_0=\Theta(q^2p)$.
Consequently, the final estimates are $(1+\eps')^{\Theta(q^2p)}=(1+\eps')^{\Theta(\log^3{n})}$-approximate.

Recall from Section~\ref{sec:bernstein}, that given hub sets, Bernstein's algorithm
runs in 
\begin{align*}
  O\left(\sum_{i=0}^{q-1} a_i\cdot d_{i+1}m\log{n}\log{(nW)}/\eps'\right)&=O\left(nm\log^2{n}\log{(nW)}\sum_{i=0}^{q-1}\frac{a_i}{a_{i+2}}/\eps'\right)\\
  &=O(nm\log^5{n}\log{(nW)}/\eps')
\end{align*}
total time.
By setting $\eps'=\Theta(\eps/\log^3{n})$, we get
$O(nm\log^8{n}\log{(nW)}/\eps+\Delta)$ total update time
with high probability.
We have thus proved the following theorem.
\begin{theorem}\label{t:bernstein}
Let $G$ be a weighted directed graph. There exists a partially-dynamic
  algorithm maintaining $(1+\eps)$-approximate distance estimates between all
  pairs of vertices of $G$ in
  $O(mn\log^{8}{n}\log{(nW)}/\eps+\Delta)$
  total time.
The algorithm is Las Vegas randomized and assumes an oblivious adversary.
\end{theorem}

\end{document}